\documentclass{article}

\usepackage{tikz}
\usepackage{float}
\usepackage[numbers]{natbib}

 \def\BIBand{and}%
 \bibpunct[, ]{[}{]}{,}{n}{}{,}%

\usepackage{dsfont}
\usepackage{mathtools}
\usepackage{amsmath}
\usepackage{amssymb}
\usepackage{amsthm}
\usepackage{algorithm,algorithmic}
\usepackage{todonotes}

\def\<{{\langle}}
\def\>{{\rangle}}

\def\1Lip{1\text{-Lip}}
\def\l({\left(}
\def\r){\right)}

\def\dd{\mathrm{d}}

\DeclareMathOperator{\dom}{dom}

\DeclareMathOperator{\Pb}{P}
\DeclareMathOperator*{\esssup}{ess\,sup}
\DeclareMathOperator*{\argmax}{argmax}
\DeclareMathOperator*{\argmin}{argmin}

\usepackage[colorlinks=true,breaklinks=true,bookmarks=true,urlcolor=blue,
     citecolor=blue,linkcolor=blue,bookmarksopen=false,draft=false]{hyperref}

\def\EMAIL#1{\href{mailto:#1}{#1}}

\theoremstyle{plain}
\newtheorem{theorem}{Theorem}
\newtheorem{lemma}{Lemma}
\newtheorem{proposition}{Proposition}
\newtheorem{corollary}{Corollary}
\theoremstyle{definition}
\newtheorem{definition}{Definition}

\theoremstyle{remark}
\newtheorem{remark}{Remark}

\newenvironment{APPENDICES}{\appendix}{}

\title{Decarbonization of financial markets:\\ a mean-field game approach}
\author{
Pierre Lavigne\\
Univesrit\'e C\^ote d'Azur, Parc Valrose, 06108 Nice, France\\
\EMAIL{pierre.lavigne@unice.fr}
\and
Peter Tankov\\
CREST-ENSAE, Institut Polytechnique de Paris, \\5 avenue Henry Le Chatelier 91120 Palaiseau, France\\
\EMAIL{peter.tankov@ensae.fr}
}
\date{}
\begin{document}
\maketitle
\begin{abstract}
We develop a financial market model in which a large population of firms chooses dynamic emission strategies under climate transition risk, interacting with both environmentally concerned and neutral investors. Firms face a trade-off between financial returns and environmental performance, while their decisions are coupled through an equilibrium stochastic discount factor determined by investors’ portfolio allocations. The framework is formulated as a mean-field game, for which we establish existence and uniqueness of a Nash equilibrium among firms. We propose a convergent numerical scheme to compute the equilibrium and use it to study how climate transition risk and green-minded investors affect decarbonization dynamics and asset prices. Our results show that uncertainty about future climate risks and policies increases aggregate emissions and widens valuation spreads between green and brown firms. Although environmentally concerned investors can partially offset these effects by raising the cost of capital for high-emission firms and incentivizing emission reductions, policy uncertainty weakens their impact. Even a large share of green-minded investors is insufficient to reverse emission growth when future climate policies are unclear, highlighting the crucial role of credible and predictable climate policy in enabling financial markets to support decarbonization.
\end{abstract}


\section{Introduction}

Decarbonization of industry is an essential ingredient for a successful environmental transition, and the financial sector has a key role to play in meeting the financing needs of green companies and directing the funds away from brown, carbon intensive projects. The amount of assets invested in climate-aware funds increased more than two-fold in each year between 2018 and 2021, reaching USD 408 billion at the end of 2021 \citep{morningstar}, and several authors aimed to quantify the impact of these additional funding flows on the emission reductions in the real economy. Such impact can be achieved only if green-minded investors target a sufficiently large proportion of companies \citep{berk2021impact}, and the environmental performance of each company depends on factors which are not directly controlled by investors, such as the general economic situation, financial health of the company, and future climate policies. The decarbonization of a financial market is therefore the result of interaction of a large number of companies, operating in an uncertain environment, and should be modeled as a dynamic stochastic game with a large number of players.  

Here we develop a dynamic model for the decarbonization of a large financial market, arising from an equilibrium dynamics involving companies and investors, and built using the analytical framework of mean-field games. 
Mean-field games, introduced in \cite{huang2006large} and  \cite{lasry2007mean} provide a rigorous way to pass to the limit of a continuum of agents in stochastic dynamic games with a large number of identical agents and symmetric interactions. In the limit, the representative agent interacts with the average density of the other agents (the mean field) rather than with each individual agent. This limiting argument simplifies the problem, leading to explicit solutions or efficient numerical methods for computing the equilibrium dynamics.

The key ingredient of our framework is the notion of mean-field financial market, which describes a large financial market with a continuum of small firms, where the performance of each firm is driven by idiosyncratic noise and a finite number of market-wide risk factors (common noise). We assume that the investors in this market are 'large' meaning that in every investor's portfolio the idiosyncratic risk of small firms is completely diversified, and the portfolio value depends only on market-wide risk factors. Consequently, and consistently with the classical finance theories \citep{ang2007capm,jagannathan1995capm}, only market-wide risk factors are priced, and the stochastic discount factor depends only on the common noise and the 'mean-field'. 

We then consider a mean-field market where shares of a continuum of carbon-emitting firms are traded. Each firm determines its dynamic stochastic emission schedule based on its own information and on the market-wide risk factors and market-wide decarbonization dynamics, rather than on the individual decision of each other small firm, which it cannot observe. To fix its emission level, each firm optimizes a criterion depending on its financial and environmental performance. The financial performance is measured by the market value of the firm's shares and therefore depends on the stochastic discount factor, introducing an interaction between the firms. The environmental performance is measured by carbon emissions, which are penalized in the optimization functional of the firm. The strength of this emission penalty is stochastic, reflecting the uncertainty of climate transition risk. This ``stochastic carbon penalty" is a key feature of our model, allowing us to analyze the impact of climate policy uncertainty on market decarbonization and asset prices in a diffusion setting. We show that  higher uncertainty about future climate policies and
transition risks creates incentive for all companies to emit more carbon and leads to
higher share prices and higher spreads between share prices of carbon efficient and
carbon intensive companies, confirming the findings of \cite{de2022climate} in a more realistic setting with stochastic emission schedules. 

The second key ingredient of our model is the interaction between two large investors (or two classes of investors), with different views about the future: while the regular investor uses the real-world measure, the green-minded investor uses an alternative measure, which may, for example, overweight the probability of some environmental policies, making the costs of climate transition more material. In the presence of such green-minded investors, all companies will reduce their emissions and pay lower dividends, leading to lower share prices. However, carbon intensive companies are affected much stronger than climate-friendly carbon efficient companies. This pressure on share prices, in turn, spurs the polluting companies to decrease their emissions.

We summarize the interaction channels and the structure of the game of the present article in figure \ref{fig:structure} below. The interaction goes as follows:
\begin{itemize}
    \item On the one hand, given a stochastic discount factor $\xi$, the firms choose optimal emissions $\psi$, driving their economic values $V$ (the problem of the firms is presented on page \pageref{eq:criterion});
    \item On the other hand, the investors $i \in \{r,g\}$ optimize their wealth $W^i_T$ depending on their greenness (the investors' problems are presented on page \pageref{pb:investor-individual-r});
    \item All the players (the firms and the investors) are coupled through the terminal market clearing condition: the wealth of the investors equals the economic value of the firms (the market clearing condition is presented on page \pageref{clearing}).
\end{itemize}

\begin{figure}[H]
    \centering 
	    \begin{tikzpicture}[line/.style={shorten >=0.2cm,shorten <=0.2cm},thick,scale = 0.7, every node/.style={transform shape}]

	        
	        \node (firm) [rectangle,draw,inner sep=0pt,minimum width = 4.5cm, minimum height = 2.5cm,text depth=5ex] at (-6.5,0) {Firm's problem};
            \node (J) at (-6.5,-0.4) {$\displaystyle  \sup_{\psi} J[\xi](\psi)$};
            
            \node (investor) [rectangle,draw,inner sep=0pt,minimum width = 4.5cm, minimum height = 2.5cm,text depth=5ex] at (6.5,0) {Investor's problem};
            \node (G) at (6.5,-0.4) {$\displaystyle \sup_{W^i_T} U^i[\xi](W^i_T), \; i \in \{r,g\}$};
            
            \node (MC) [rectangle,draw,inner sep=0pt,minimum width = 4.5cm, minimum height = 2.5cm,text depth=5ex] at (0,0) {Market Clearing};
            \node (eq) at (0,-0.4) {$\displaystyle \mathbb{E}[V_T\vert \mathcal{F}_T^0] = W^r_T + W^g_T$};

            \node (psi)  at (-3.25,2.5) {Emissions $\psi$ and firm value $V$};
            \node (W)  at (3.25,2.5) {Wealths $(W^r_T, W^g_T)$};
            \node (xi-firm)  at (-3.25,-2.5) {Stochastic discount factor $\xi$};
            \node (xi-investor)  at (3.25,-2.5) {Stochastic discount factor $\xi$};


            \path [black,bend left,line] [->]  (MC) edge (firm);
            \path [black,bend left,line] [<-]   (investor) edge (MC) ;

            \path [black,line,bend left] [->]  (firm) edge (MC);
            \path [black,bend left,line] [<-]  (MC) edge  (investor);
        \end{tikzpicture}
    \caption{Structure of the game.}
    \label{fig:structure}
\end{figure}

We rigorously prove the existence and uniqueness of the mean-field game Nash equilibrium for the contunuum of firms interacting through market prices of their shares, providing a robust solution to the stochastic ``decarbonization game" in a competitive environment. The equilibrium is materialized by the equilibrium stochastic discount factor, which can be used to compute share prices and emission strategies for each firm. We then develop a convergent numerical algorithm to compute the equilibrium and use it to study the impact of climate transition risk and green investors on the market decarbonization dynamics and share prices. 

The paper is structured as follows. In Section \ref{mfg.sec}, we describe the mean-field game setting and define the optimization problem of individual firms and the notion of equilibrium used in this paper. In Section \ref{sec:main-result}, we state and prove the main existence and uniqueness result, propose a convergent numerical algorithm for computing the equilibrium, and prove its convergence to the solution of the mean field game problem is shown. Section \ref{examples.sec} illustrates the theory with numerical examples and discusses the implications of our results, and section \ref{conclusion.sec} concludes the paper.
In Appendix \ref{app:nfirm.sec} we link the mean field game model with a $n$-player game, and present a convergence result.

\paragraph{Related literature}

Our paper contributes to the emerging literature on impact investing, and on the role of climate risk and uncertainty in sustainable finance. \cite{pastor2021sustainable} build a one-period equilibrium model to describe the impact of investors' ESG preferences on the performance of green assets. \cite{avramov2022sustainable} analyze the asset pricing implications of the uncertainty of corporate ESG profiles, also in a one-period setting. The impact of sustainable investing on the cost of capital, firm behavior 
 and social outcomes is studied, in a one-period equilibrium model, in an early contribution of \cite{heinkel2001effect}, and more recently, theoretically and empirically, in \citep{berk2021impact,chowdhry2019investing,green2021allocation,landier2020esg,oehmke2022theory}. The impact of climate risk on investors' choices is also studied in a number of papers. \cite{bolton2021investors} provide empirical evidence that investors are requiring compensation for their exposure to carbon risk; \cite{ilhan2021carbon} argue that climate policy uncertainty is priced in the options market, and in particular that cost of downside protection is larger for firms with carbon intensive business models; \cite{bourgey2022bridging} quantify the impact of carbon emissions on a firm's credit risk in the context of shared socio-economic pathways, \cite{huang2018impact} show that firms adapt their financing choices to physical climate risks affecting their performance; and \cite{krueger2020importance} report the results of a survey showing that climate risks are considered to be material by institutional investors, while \cite{alekseev2022quantity}, \cite{andersson2016hedging} and  \cite{engle2020hedging}  develop investment strategies allowing to hedge these risks. 

Given that the vast majority of theoretical research on impact investing is done in a one-period setting,  \cite{de2022climate} develop a continuous-time model for the impact of green investment on company emission dynamics, allowing for time-dependent emission schedules. However, in their model the emissions schedules are deterministic (fixed at time $0$), and the arguments used to prove existence of equilibrium are partly heuristic, because the number of firms is assumed to be finite. In our model the emission schedules are stochastic and the firms can change them at any time depending on their financial performance of the individual company, the market dynamics, and the materialization of climate risk;  moreover, the existence and uniqueness of equilibrium is shown rigorously. 

Our model is based on the general framework of mean field games with common noise \citep{ahuja2016wellposedness,cardaliaguet2019master,carmona2018probabilistic,carmona2016mean,djete2021large}. The latter are often untracktable due to the lack of compactness of the space containing the random measures, which correspond to the conditional law of the state process given the common noise. A possible workaround is to discretize the common noise to recover compactness and apply Kakutani's or Schauder's fixed point theorem \citep{barrasso2022controlled,cardaliaguet2022mean,carmona2016mean}.
In this paper, instead of working on the space of random probability measures, we write our mean field game problem as a fixed point equation for the stochastic discount factor (interaction term). The proof strategy is to form a strictly convex minimization problem whose first order condition is equivalent to the fixed point problem. This allows us to show the existence and uniqueness of a strong solution to the mean field game problem and to apply the generalized conditional gradient algorithm \citep{bredies2009generalized} to provide a numerical solution.

We also contribute to the literature on equilibrium models based on mean-field games.  In mathematical finance  \cite{casgrain2020mean,feron2022price,fu2021mean,fujii2022mean,gomes2021mean}, among other authors, developed models of equilibrium price formation based on interaction of a large number of traders in the framework of mean-field games. In these models, although market clearing may be imposed, the agents are cost minimizers rather than utility optimizers so that Arrow-Debreu equilibrium is not constructed. A game of many utility maximizing agents is considered in \cite{lacker2020many}, but the agents are price takers in this reference. Instead, we build on the classical equilibrium approach, along the lines of \cite{duffie1986stochastic,duffie1985implementing,huang1987intertemporal} and many more recent papers \citep{anderson2008equilibrium,hugonnier2012endogenous,riedel2013existence}, where utility-optimizing agents determine prices in equilibrium. Although mean-field games and related notions have been used in economics, for example, in industry dynamics models \citep{luttmer2007selection} and for modeling income and wealth distributions \citep{achdou2022income}. They have also been used in equilibrium pricing theory, see \citep{fujii2022mean,fujii2022strong}.

\section{A mean-field game model of decarbonization}
\label{mfg.sec}

In this section we present the problem under study.  We start this section presenting the stochastic context of the model. We then present the problem of the representative firm. We continue with the problem of the green and brown investors. Finally, we state the mean field game equilibrium problem we solve in this article.
This problem can be understood as a limiting game for a $N$-firm problem, as $N$ tends to infinity. Such model is presented, and the connexion with the mean field limit is established, in Appendix \ref{app:nfirm.sec}.

\paragraph*{Stochastic context and notations}The following notations shall be used for the rest of this paper. Let $T>0$ be a time horizon. 
Let $(\Omega,\mathcal{F}, \mathbb{P})$ be a probability space, supporting a pair $(B,B^0)$ of independent (possibly multidimensional) standard Brownian motions and an independent random variable $X$ corresponding to the initial condition. We denote by $\mathbb F = (\mathcal F_t)_{0\leq t\leq T}$ the $\mathbb P$-complete natural filtration of $(X,B,B^0)$ and by $\mathbb F^0 = (\mathcal F^0_t)_{0\leq t\leq T}$ the $\mathbb P$-complete natural filtration of $B^0$. 
In our setting, $B$ is an idiosyncratic noise and $B^0$ is the common noise. Following standard mean field game theory, the common noise $B^0$ appears in the state equation of the representative firm and captures the intrinsic randomness of the market environment. The expectation with respect to the reference measure $\mathbb P$ will be denoted simply by $\mathbb E$, and the expectation with respect to any other measure $\mathbb Q$ will be denoted by $\mathbb E_{\mathbb Q}$. 

In general, for a $\sigma$-field $\mathcal G$, we denote by $L^p(\mathcal G)$ the space of $\mathcal G$-measurable random variables $X$ satisfying
$$
\mathbb E[|X|^p]<+\infty,
$$
and for a filtration $\mathbb G = (\mathcal G_t)_{0\leq t\leq T}$ we denote by $L^p(\mathbb G)$ the space of $\mathbb G$-progressively measurable random processes $(X_t)_{0\leq t\leq T}$ satisfying
$$
\mathbb E\left[\int_0^T |X_t|^p \dd t\right]<+\infty.
$$
Similarly, we denote by $L^0(\mathcal G)$ and $L^0(\mathbb G)$ the spaces of measurable random variables and progressively measurable random processes without a specific integrability condition. Moreover, for a $\sigma$-field $\mathcal G$ we write
\begin{equation*} 
    L_+^p(\mathcal G) = \{X \in L^p(\mathcal G); X\geq 0\;\text{a.s.} \}.
\end{equation*}
Finally,  $\mathds{1} \in L^p(\mathcal{F}_T)$ for all $p$ denotes the constant variable equal to $1$ almost surely.

It is clear that $\mathbb F^0$ satisfies the \emph{immersion property} with respect to $\mathbb F$  \citep[page 5]{carmonadelaruev2}, which implies, in particular, that for all $\zeta\in L^1(\mathcal F_t)$, $\mathbb E[\zeta|\mathcal F^0_T] = \mathbb E[\zeta|\mathcal F^0_t]$.

\subsection{The representative firm problem}
\label{sec:representative-firm}

We consider an economy with a large number (continuum) of firms, whose shares are traded in the financial market.  
The value of the representative firm per share  is assumed to evolve according to the following stochastic differential equation
\begin{equation}
    \dd  V_t =  V_t (\mu_t \dd t + \sigma_t \dd B_t + \sigma^0_t  \dd B^0_t) + e_t \psi_t \dd t,\quad  V_0 = V, \label{eq:state-equation}
\end{equation}
where $V\in L^2(\mathcal F_0)$ and {$V>0$\;almost surely}. 
Here, the Brownian motion $B$ models the idiosyncratic risk of the company, while $B^0$ corresponds to market-wide and economy-wide risk factors. The process $\psi_t$ denotes the instantaneous emissions of the representative firm and the process $e_t$ is inversely proportional to the emission intensity of production and will be referred to as ``emission efficacy" in the sequel; the values of this process are low for carbon-intensive companies and high for green, carbon efficient firms which create a lot of wealth with little emissions. Thus, in the absence of emissions, the firm value follows an autonomous dynamics (with perhaps $\mu_t<0$), and to increase the firm value beyond this baseline, the firm must produce and therefore increase its emissions.

In accordance with the mean-field games paradigm, the process $V$ defines simultaneously the dynamics of a single firm value and the distribution of values of all firms in the market at any given time. For example, the law of $V_t$ conditional on the common noise, denoted by $m_t = \mathcal{L}(V_t|\mathcal F^0_t)$ defines the distribution of firm values at time $t$, given the realization of the common factors up to time $t$. The process $(m_t)_{0\leq t\leq T}$ is a random process, adapted to the filtration of the common noise $\mathbb F^0$ and taking values in the space of probability measures on $\mathbb R_+$. Since $m_t$ is a probability measure, we implicitly assume that the total quantity of all firms in the market (the market size) is normalized to $1$. 

The processes $(\mu_t)_{0\leq t\leq T}$, $(\sigma_t)_{0\leq t\leq T}$, $(\sigma^0_t)_{0\leq t\leq T}$ and $(e_t)_{0\leq t \leq T}$ are assumed to belong to $L^0(\mathbb F)$ and satisfy
$$
\int_0^T (|\mu_t| + |\sigma_t|^2 + |\sigma^0_t|^2 + |e_t|^2) \dd t<+\infty\quad \text{a.s.}
$$
The parameter $\mu$ is real-valued, $\sigma$ and $\sigma^0$ are vector-valued of the same dimension as the corresponding Brownian motions, and $e$ takes positive values almost surely. 

Each firm controls its value via the control $\psi$ which represents its emission flow.
We denote $V^\psi$ the solution to the state equation for a given control $\psi$.
We define the following stochastic exponential 
\begin{equation} \label{def:E-s-t}
    \mathcal{E}_t =  \exp\left(\int_0^t \sigma_s \dd B_s + \int_0^t \sigma^0_s \dd B_s^0  + \int_0^t (\mu_s - \frac{|\sigma_s|^2}{2} - \frac{|\sigma_s^0|^2}{2} )\dd s\right), \\
\end{equation}
and we denote $\mathcal{E}_{s,t} \coloneqq \mathcal{E}_{t} \mathcal{E}_{s}^{-1}$ for all $0 \leq s \leq t \leq T$.
\begin{lemma}
    For any $\psi$ such that $\int_0^T \psi_s^2 ds <+\infty$ a.s., the equation \eqref{eq:state-equation} has a unique solution  given by 
    \begin{equation} \label{eq:state-equation-explicit-form}
        V^\psi_t = \mathcal{E}_t V + \int_0^t \mathcal{E}_{s,t} e_s \psi_s \dd s.
        \end{equation}
\end{lemma}
\begin{proof}
This result follows from \cite[Theorem 7, Chapter V]{protter2005stochastic} and \cite[Theorem 52, Chapter V]{protter2005stochastic}.
\end{proof}


To compute share prices, we assume that there exists a \emph{stochastic discount factor} process $(\xi_t)_{0\leq t\leq T}$ with $\xi_0 =1$ and $\xi_t > 0$ for all $t\in[0,T]$ a.s., such that
\begin{enumerate}
\item  $(\xi_t)_{0\leq t\leq T}$ is a $\mathbb F^0$-martingale;
\item  $(S_t\xi_t)_{0\leq t\leq T}$ is a $\mathbb F$-martingale, where $(S_t)$ is the share price of the representative firm. 
\end{enumerate}
The stochastic discount factor process is adapted to the filtration of the common noise $\mathbb F^0$, because, in accordance with the financial theory, we assume that \emph{idiosyncratic risks} are not priced in the market since they are fully diversifiable.

As in many papers on equilibrium asset pricing, we assume that at terminal time $T$, the firm pays to its shareholders a terminal liquidating dividend per share equal to the  value of the firm.
The share price of the firm is therefore given by
\begin{equation} \label{eq:price-equation}
    S^{\xi,\psi}_t = \frac{1}{\xi_t} \mathbb{E} \left[ \left. \xi_T V^\psi_T\right \vert \mathcal{F}_t \right],
\end{equation}
where we make explicit the dependence of the share price on the discount factor and the emission strategy of the firm. The stochastic discount factor process will be determined endogenously in equilibrium. In view of the martingale property, we will use interchangeably the stochastic discount factor process $(\xi_t)_{0\leq t\leq T}$ and the stochastic discount factor $\xi\in L^1_+(\mathcal F^0_T)$, with $\xi_t = \mathbb E[\xi|\mathcal F^0_t]$.


 We assume that carbon emissions are associated with direct or indirect penalty for the firm, which affects the decision-making by the firm but does not appear in its share price.  This penalty may correspond, for example, to (i) reputational damages for the firm, (ii) potential climate risks, which are not yet priced in the markets but are taken into account by the management, (iii) managerial incentives linked to environmental performance. The total penalty accrued between $0$ and $T$ is given by
$$
\int_0^T \alpha_t c(\psi_t) \dd t,
$$
where $c \colon \mathbb{R} \to \mathbb{R}$ is a cost function.  We assume that $c$ is a strictly convex even differentiable function, which satisfies
$$
c(0) = 0\quad \text{and}\quad \frac{c(x)}{x}\to+\infty\quad \text{as}\quad x\to \infty.
$$
The stochastic process $(\alpha_t)_{0\leq t\leq T} \in L^0(\mathbb F)$ is strictly positive models the future intensity of such costs.
We define $c^* \colon \mathbb{R} \to \mathbb{R}$ the Fenchel transform of $c$ and the function $I \colon \mathbb{R} \to \mathbb{R}$
$$
c^*(x^*) = \sup_{x} \{xx^* - c(x)\}, \qquad I(x^*) \coloneqq \argmax_{x} \{x x^* - c(x)\}.
$$
We assume that $c^*$ is strongly convex and admits a bounded second order derivative.
As cost function, we may consider for example the quadratic cost 
$$
c(x) = \frac{|x|^2}{2},\qquad c^*(x^*) = \frac{|x^*|^2}{2},\qquad I(x^*) = x.
$$

The representative firm objective is to maximize its share price net of emission penalty. We define the criterion of a representative firm
\begin{equation} \label{eq:criterion}
    J[\xi](\psi) = \mathbb{E}\left[S^{\xi,\psi}_0 - \int_0^T \alpha_t c(\psi_t) \dd t \right] =\mathbb{E} \left[ \xi V^\psi_T - \int_0^T \alpha_t c(\psi_t) \dd t \right].
\end{equation}

Each company therefore aims to achieving a trade-off between maximizing financial performance (share price) and minimizing the negative environmental impact, as measured by the emission penalty. Using equilibrium share price as the optimization objective of the firm is consistent with the literature \citep{heinkel2001effect}. In \cite{de2022climate}, the authors suggest to use the expected sum of discounted future prices, but since the price at time $t=0$ already accounts for all future cash flows, it seems sufficient and more natural to use the initial price. On the other hand, convex emission penalty similar to ours was used, for example in the context of a single-firm optimization problem in \cite{bourgey2022bridging}. 
The optimization problem of the firm is defined as follows
\begin{equation} \label{pb:representativ-firm} \tag{$\mathrm{P_f}$}
    \esssup_{\psi \in \mathcal{C}} J[\xi](\psi),\quad \mathcal{C} \coloneqq \left\{\psi \in L^0(\mathbb{F}), \; \; \mathbb E\left[\int_0^T \alpha_t c(\psi_t)dt\right] < + \infty \right\}.
\end{equation}

\begin{lemma}\label{firm.lm}
Assume that $\mathbb E[\xi V \mathcal E_T]<+\infty$ and 
$$
\mathbb E\left[\int_0^T \alpha_t c^*\left(\frac{e_t}{\alpha_t} \mathbb E[\xi \mathcal E_{t,T}|\mathcal F_t]\right) \dd t\right]<+\infty.
$$
Then, the unique optimal solution of the representative firm problem is given by 
\begin{align} 
\psi^\xi_t := I\left(\frac{e_t}{\alpha_t}  \mathbb{E}[\xi\mathcal{E}_{t,T}|\mathcal F_t]\right).\label{psiopt}
\end{align}
\end{lemma}

\begin{proof}
Using the formula \eqref{eq:price-equation} for the share price and the formula \eqref{eq:state-equation-explicit-form} for the terminal firm value,  we see that for any $\psi \in \mathcal C$, the optimization functional of the representative firm satisfies
\begin{align}
    J[\xi](\psi) &= \mathbb E[\xi V \mathcal E_T] + \mathbb{E}\left[   \int_0^T  \xi\mathcal{E}_{t,T} e_t  \psi_t \dd t -   \int_0^T \alpha_t c(\psi_t)  \dd t \right]. \label{Jfunc.eq}
\end{align}
By convex duality, using that $\alpha$ is strictly positive a.s.,
\begin{align*}
    \int_0^T \mathbb E[\left|\xi\mathcal{E}_{t,T} e_t  \psi_t\right|] \dd t & = \int_0^T \mathbb E[ \mathbb E[\xi\mathcal{E}_{t,T}|\mathcal F_t] e_t  |\psi_t|] \dd t \\ & \leq \int_0^T\mathbb E\left[ \alpha_t c(\psi_t)\right] \dd t +\int_0^T \mathbb E\left[\alpha_t c^*\left(\frac{e_t}{\alpha_t} \mathbb E[\xi\mathcal{E}_{t,T}|\mathcal F_t] \right)\right] \dd t <+ \infty
    & 
\end{align*}
Thus, by Fubini's theorem,
\begin{align*}
    J[\xi](\psi) &= \mathbb E[\xi V \mathcal E_T] + \mathbb{E}\left[   \int_0^T  \mathbb E[\xi\mathcal{E}_{t,T}|\mathcal F_t] e_t  \psi_t \dd t -   \int_0^T \alpha_t c(\psi_t)  \dd t \right].
\end{align*}
Then, once again by convex duality,
\begin{align*}
    J[\xi](\psi) &\leq \mathbb E[\xi V \mathcal E_T]+\mathbb E\left[\int_0^T \alpha_t c^*\left(\frac{e_t}{\alpha_t} \mathbb E[\xi \mathcal E_{t,T}|\mathcal F_t]\right) \dd t\right] ,
\end{align*}
and the equality is achieved for the optimal value $\psi^\xi$ given in \eqref{psiopt}, concluding the proof.

\end{proof}

\begin{remark}\label{emission.rem}
The optimal emissions level \eqref{psiopt} is given by the transformation function $I$, which depends on cost function of emissions, applied to a ratio of 
$$
e_t \mathbb{E}[\xi\mathcal{E}_{t,T}|\mathcal F_t],
$$
which corresponds to the net present value at time $t$ of the wealth generated at time $T$ from one unit of emissions at time $t$ and $\alpha_t$ which determines emission penalty. The firm will therefore emit more if the wealth it generates for investors has more value in the market at current prices  and less if the emissions are more strongly penalized.
 
\end{remark}

\subsection{The investors problem}
\label{investors.sec}

We assume that in the market there are two competing investors: a regular and a green investor. The two investors may have different views on the future evolution of the firm values and emission penalties. For example, the green investor may have a different view on the future distribution of the emission cost parameter $\alpha$ because this investor considers future climate risks to be more material. 

Each investor holds a large diversified portfolio, where the idiosyncratic risk of individual stocks is fully diversified, so that 
$\mathcal W$, the space of admissible values for the terminal portfolio of each investor is a subset of $\mathcal F^0_T$, the common noise $\sigma$-field. Moreover, we assume that the market is complete, meaning that every terminal wealth value $W\in \mathcal W$ can be replicated by each investor using the available traded instruments.

Given a stochastic discount factor $\xi\in L_+^1 (\mathcal F^0_T)$ with $\mathbb E[\xi]=1$ we therefore define the space of admissible values for the terminal wealth of each investor by
$$
\mathcal W = \{W \in \mathcal F^0_T : \mathbb E[\xi |W|]<+\infty\}. 
$$

Following the standard macroeconomic equilibrium construction, we assume that the investors dispose of initial endowements with pay-offs $X^r \in \mathcal W$ for the regular investor and $X^g\in \mathcal W$ for the green investor.  
Each investor maximizes CARA utility function of terminal wealth under the budget constraint
\begin{align} \label{pb:investor-individual-r} \tag{$\mathrm{P_r}$}
 \sup_{W \in \mathcal W} \, 1 - \mathbb{E} \left[ e^{-\gamma^r  W} \right], \quad \mathrm{s.t.} \quad \mathbb{E} \left[ \xi W\right] \leq \mathbb{E}\left[ \xi X^r\right], \\ 
 \sup_{W \in \mathcal W} \, 1 - \mathbb{E}_{\mathbb{P}^g} \left[e^{-\gamma^g  W} \right], \quad \mathrm{s.t.} \quad \mathbb{E} \left[ \xi W \right] \leq \mathbb{E}\left[ \xi X^g\right], \label{pb:investor-individual-g} \tag{$\mathrm{P_g}$}
\end{align}
where $\gamma^r$ and $\gamma^g$ are risk aversion coefficients of, respectively, the regular and the green investor, and $\mathbb P^g$ is the subjective probability measure of the green investor. We assume that $\mathbb P^g$ is equivalent to $\mathbb P$ and that the Radon-Nikodym density
\begin{equation*}
    Z:= \frac{\dd \mathbb P^g}{\dd \mathbb P}
\end{equation*}
satisfies $Z\in \mathcal F^0_T$. This is justified by the fact that the green investor holds a fully diversified portfolio and therefore does not care about the idiosyncratic risk of individual stocks. We make no further assumptions about the structure of $Z$ at this point but refer the reader to section  \ref{examples.sec} for an explicit construction.




\begin{lemma}\label{investor.lm}
Assume that the stochastic discount factor satisfies $\mathbb E[\xi \ln \xi ]<\infty$ and $\mathbb E[\xi \ln \frac{\xi}{Z}]<+\infty$.
Then the optimal terminal wealth values for the regular and the green investor are given by
$$
W^r = \mathbb{E}\left[ \xi X^r\right] + \frac{1}{\gamma^r} \mathbb E[\xi \ln \xi]- \frac{1}{\gamma^r} \ln \xi ,\qquad W^g = \mathbb{E}\left[ \xi X^g\right]+ \frac{1}{\gamma^g} \mathbb E[\xi \ln (\xi/Z)] - \frac{1}{\gamma^g} \ln (\xi/Z) .
$$
\end{lemma}

\begin{proof}
Recall that by Fenchel duality, for any $x> 0$ and $y\in \mathbb R$, we have 
\begin{equation} \label{ineq:exponential}
    e^y \geq xy - x \ln (x) + x,
\end{equation} and equality holds at $x = e^y$. Consider the optimization problem of the green investor \eqref{pb:investor-individual-g}. Let $\kappa \geq 0$. Using inequality \eqref{ineq:exponential} for $x = \kappa \xi Z^{-1}(\omega)$ and $y = -\gamma^g W(\omega)$ for all $W \in \mathcal W$ and $\omega \in \Omega$, yields
$$
\mathbb E \left[Z e^{-\gamma^g W} \right] \geq \mathbb E\left[-\kappa \xi \gamma^g W - \kappa \xi \ln (\kappa \xi) + \kappa \xi + \kappa \xi \ln Z \right].
$$
Then using the budget constraint $\mathbb{E} \left[ \xi W \right] \leq \mathbb{E}\left[ \xi X^g\right]$ and the fact that $\mathbb E[\xi] = 1$ yields
\begin{align*}
\mathbb E \left[Z e^{-\gamma^g W} \right] \geq \left(-\gamma^g \mathbb{E}\left[ \xi X^r\right] - \mathbb E[\xi \ln (\xi/Z) ]\right)\kappa -\kappa \ln \kappa + \kappa.
\end{align*}
Applying Fenchel duality once again, we obtain
$$
\mathbb E\left[Z e^{-\gamma^g W}\right] \geq \sup_{\kappa>0} \{(-\gamma^g \mathbb{E}\left[ \xi X^r\right] - \mathbb E[\xi \ln(\xi/Z)])\kappa -\kappa \ln \kappa + \kappa\} = e^{-\gamma^g \mathbb{E}\left[ \xi X^r\right] - \mathbb E[\xi \ln (\xi/Z)]},
$$
where the supremum is attained for $\kappa^* = e^{-\gamma^g \mathbb{E}\left[ \xi X^r\right] - \mathbb E[\xi \ln (\xi/Z)]}$. 
Now, consider $W^g = -\frac{1}{\gamma^g} \ln (\kappa^* \xi/Z)$. Since $x\ln x$ is bounded for $x<1$, it follows that $W \in \mathcal W$. Since $\kappa^* \xi/Z = e^{-\gamma^g W^g}$, we get that 
\begin{align*}
\mathbb E[Ze^{-\gamma^g W^g}]  &= \mathbb E[-\kappa^* \xi \gamma^g W^g - \kappa^* \xi \ln (\kappa \xi) + \kappa^*\xi + \kappa^* \xi \ln Z]\\
& = \mathbb E[\kappa^* \xi] = \kappa^* = e^{-\gamma^g \mathbb{E}\left[ \xi X^r\right] - \mathbb E[\xi \ln \xi/Z]},
\end{align*}
and we conclude that $W^g$ is the optimal solution for the green investor. The solution for the regular investor then follows by taking $Z = \mathds{1}$. 
\end{proof}

\subsection{Market equilibrium} \label{sec:MC}
We assume that the initial endowements of the agents $X^r$ and $X^g$  satisfy the market clearing condition:
$$
X^g + X^r = \mathbb{E} \left[ V^\psi_T \vert \mathcal{F}^0_T \right],
$$
in other words, together the two investors hold the entire market. 
In the right-hand side, recalling that the law of $V^\psi_T$ conditional on the common noise defines the distribution of firm values at time $T$, it is clear that $\mathbb{E} \left[ V^\psi_T \vert \mathcal{F}^0_T \right]$ is the aggregate value of all firms in the market at time $T$.

Following the standard economic theory we say that the market is in equilibrium if the stochastic discount factor $\xi$ is such that the optimal holdings of the agents also satisfy the market clearing condition, that is, 
\begin{equation} \nonumber
    W^g + W^r = \mathbb{E} \left[ V^\psi_T \vert \mathcal{F}^0_T \right].
\end{equation}

Substituting the expressions of Lemma \ref{investor.lm} into this equation, and rearranging the terms, we get: 
$$
\mathbb E\left[\xi\left(\gamma^*\mathbb E\left[ V^\psi_T \vert \mathcal{F}^0_T \right]+ \ln \xi - \rho\ln Z\right)\right] = \gamma^*\mathbb E\left[ V^\psi_T \vert \mathcal{F}^0_T \right]+ \ln \xi - \rho\ln Z
$$
with $\frac{1}{\gamma^*} = \frac{1}{\gamma^r} + \frac{1}{\gamma^g}$. 
Solving for $\xi$, we finally obtain: 
\begin{equation}
    \label{clearing}
    \xi  = \frac{\exp\left(\rho \ln (Z)-\gamma^* \mathbb{E} \left[ V^\psi_T \vert \mathcal{F}^0_T \right]\right)}{\mathbb E\left[\exp\left(\rho\ln (Z)-\gamma^* \mathbb{E} \left[ V^\psi_T \vert \mathcal{F}^0_T \right]\right)\right]},
\end{equation}
In the above expression, $\rho \coloneqq \frac{\gamma^r}{\gamma^g + \gamma^r}\in(0,1)$ can be interpreted as the proportion of green investors in the market. Indeed, following \cite{de2022climate}, to simplify the interpretation of the impact of green and regular investors' wealth on the variables in equilibrium, we may assume that green and regular investors have equal relative risk aversions; that is, $\gamma^{R}=\gamma^g \mathbb E[\xi X^g]=\gamma^r\mathbb E[\xi X^r]$, where $\gamma^R$ denotes the relative risk aversion. In this case, $\rho$ is the proportion of the green investors' initial endowement, and $1-\rho$ is that of the regular investors; that is, $\rho=\frac{\mathbb E[\xi X^g]}{\mathbb E[\xi X^g]+\mathbb E[\xi X^r]}$ and $1-\rho=\frac{\mathbb E[\xi X^r]}{\mathbb E[\xi X^g]+\mathbb E[\xi X^r]}$.

\subsection{Mean-field game equilibrium of companies} \label{nash.sec}
We define $\Xi$ as the subset of all  $\xi \in L^1_+(\mathcal F^0_T)$ satisfying the assumptions of both Lemma \ref{firm.lm} and Lemma \ref{investor.lm}, in other words,
\begin{equation*}
\begin{array}{lll}
     \mathbb{E}[\xi] = 1, &  \mathbb{E}[\xi \ln (\xi)]< +\infty, &\mathbb{E}[\xi \ln \frac{\xi}{Z}] < +\infty,\\
      \mathbb E[V\xi \mathcal E_T] < +\infty, & 
\mathbb E\left[\int_0^T \alpha_t c^*\left(\frac{e_t}{\alpha_t} \mathbb E[\xi \mathcal E_{t,T}|\mathcal F_t]\right) \dd t\right]<+\infty.
\end{array}
\end{equation*}

\begin{definition}
A mean-field game equilibrium is a triple $(\bar\psi,\bar\xi,\bar m)$, where $\bar \psi \in L^2_\alpha(\mathbb F)$, $\bar \xi \in \Xi$ and $\bar m$ is a $\mathcal F^0_T$ measurable random variable valued in $\mathcal{P}(\mathbb R)$, satisfying the following conditions: 
$$
\bar{\psi} = \argmin_{\psi \in L^2_\alpha(\mathbb{F})}  J[\bar{\xi}](\psi),\qquad \bar \xi  = \frac{\exp\left(\rho \ln (Z)-\gamma^* \int x \bar m(dx)\right)}{\mathbb E\left[\exp\left(\rho\ln (Z)-\gamma^* \int x \bar m(dx)\right)\right]},
$$
and $\bar m$ is the law of $V^{\bar\psi}_T$ conditional to the $\sigma$-field $\mathcal F^0_T$ generated by the common noise, where 
\begin{equation} \label{eq:psi-V}
     V^{\bar\psi}_T = V \mathcal{E}_T +  \int_0^T e_t \mathcal{E}_{t,T} \bar\psi_t \dd t, \quad \text{and}\quad \bar\psi_t = I\left(\frac{e_t}{\alpha_t} \mathbb{E}[\bar\xi\mathcal{E}_{t,T}|\mathcal F_t]\right).
\end{equation}

\end{definition}
The first condition means that each firm solves its optimization problem subject to a given stochastic discount factor. The second condition defines the stochastic discount factor through the market clearing condition. The third condition means that this distribution of firms is obtained when every firm implements its optimal control. 

The usual approach in mean field game would be to find a fixed point in the space of $\mathcal F^0_T$-measurable random probability measures on $\mathbb R$. Due to the presence of the common noise, this space lacks compactness and the existence of equilibrium cannot be ensured. 
In our setting, the mean field game problem simplifies due to the specific form of the interaction mapping given by the market clearing equation. While the latter is fully non-linear, it only involves conditional expectation of the state process with respect to the $\sigma$-field $\mathcal F^0_T$ generated by the common noise.
It is an important simplification since we do not require the full knowledge of the law $\bar m$ to compute the stochastic discount factor. 

Then, the Nash equilibrium problem reduces to the following fixed point problem: find $\xi \in \Xi$, such that
\begin{align} \tag{FP} \label{pb:fixed-point}
\xi & =  \frac{\exp\left(\rho \ln (Z) -\gamma^* \mathbb{E}[V^\xi_T|\mathcal F^0_T] \right)}{\mathbb{E}\left[\exp\left(\rho \ln (Z) -\gamma^* \mathbb{E}[ V^\xi_T|\mathcal F^0_T] \right)\right]}.
\end{align}

Theorem \ref{main.thm}, formulated in the  next section, establishes the existence and uniqueness of solution to \eqref{pb:fixed-point}. 
It is clear by \eqref{eq:psi-V} that the uniqueness of the state price vector implies the uniqueness of the strategy and the state process. The uniqueness of the law $m^{\psi} = \mathcal L(V^\psi_T|\mathcal F^0_T)$ then holds almost surely and following the terminology of \cite[Definition 2.1]{carmona2016mean}, we obtain a unique strong mean field game solution $(\psi^\xi,\xi,m^{\psi})$.

\paragraph{Market completeness}\label{completeness} A key assumption of Section \ref{investors.sec}, required to close the loop of Nash equilibrium discussed above, is the \emph{dynamic completeness}, or, in our setting, weak completeness of the market resulting from computing stock prices 
using the formula \eqref{eq:price-equation}, with $\xi$ and $\psi$ given by the solution of the fixed-point equation \eqref{pb:fixed-point}. Dynamic completeness is in general difficult to prove in continuous-time equilibrium models, and has only been shown in specific settings \citep{anderson2008equilibrium,hugonnier2012endogenous}. Due to the complexity of our setting, a direct proof of this property seems out of reach. We therefore follow a large part of the literature by assuming \emph{a posteriori} that this property is satisfied. In practice, this means that in addition to shares of firms, in the market there are also traded options allowing to replicate any claim measurable with respect to the $\sigma$-field of the common noise $\mathcal F^0_T$. 

\section{Solution to the mean field game problem}
\label{sec:main-result}

In this section we state and prove our main results. Existence and uniqueness of the Nash equilibrium are established in Theorem \ref{main.thm}. A fixed point algorithm (Algorithm \ref{algorithm}) is provided to compute the solution of the mean field game problem and its convergence is established in Theorem \ref{thm:convergence-result}. The proof of these results is postponed to dedicated sections \ref{sec:proof-of-main-thm} and \ref{sec:proof-of-convergence} and relies on the existence and uniqueness of the solution to a convex minimization problem.  

\paragraph{Additional notations}
We define the negative entropy $h \colon \mathbb{R} \to \mathbb{R} \cup \{+\infty\}$ and the indicator function $\chi \colon \mathbb{R} \to \mathbb{R} \cup \{+\infty\}$ of the set $\{0\}$,
\begin{equation*}
    h(x) = \left\{ \begin{array}{ll}
        x (\ln (x) - 1), & \text{if} \; x > 0, \\
        0, & \text{if} \; x = 0, 
   \end{array} \right. \quad \chi(x) = \left\{ \begin{array}{ll}
       0 , & \text{if} \; x = 0, \\
       +\infty, & \text{otherwise.} 
  \end{array} \right.
\end{equation*}
Further, we define an entropic functional $H \colon L^1_+(\mathcal F^0_T) \to \mathbb{R} \cup \{+\infty\}$ and a convex functional $L \colon L^1_+(\mathcal F^0_T) \to \mathbb{R} \cup \{+\infty\}$,
\begin{equation}
    H(\xi) = \mathbb E[Z^{\rho} h(\xi/Z^{\rho})], \quad L(\xi) = \gamma^{\star}  \mathbb{E}\left[\xi V\mathcal{E}_T +\int_0^T \alpha_t c^*\left(\frac{e_t}{\alpha_t} \mathbb E[\xi \mathcal E_{t,T}|\mathcal F_t]\right) dt   \right].\label{functionals.eq}
\end{equation}
We finally define the convex functional $G \colon L^1_+(\mathcal F^0_T) \to \mathbb{R} \cup \{+\infty\}$, which will play the role of ``potential" of the fixed-point problem \eqref{pb:fixed-point}:
\begin{equation*}
    G(\xi) = H(\xi) + L(\xi) + \chi(\mathbb{E} \left[\xi \right] - 1).
\end{equation*}

\paragraph{Main results}
We state here our main results. The first theorem characterizes the solution to the fixed-point problem \eqref{pb:fixed-point}.
This results implies the existence and uniqueness of a solution to the mean field game problem. The second result establishes the convergence of the Frank-Wolfe algorithm proposed below to find the fixed point of \eqref{pb:fixed-point}.

\begin{theorem}\label{main.thm}
Assume that the parameters are such that
$L(\mathds{1})<+\infty.$
Then there is a unique solution $\bar{\xi} \in \Xi$ to the fixed-point problem \eqref{pb:fixed-point}. This solution satisfies $\bar \xi>0$ a.s.~and
\begin{equation*}
    \bar{\xi} = \argmin_{\xi \in L^1_+(\mathcal F^0_T)} G(\xi).
\end{equation*}
\end{theorem}

We now define an algorithm to solve the fixed point problem (\ref{pb:fixed-point}) based on the generalized conditional gradient algorithm. The latter is a generalization of the Frank-Wolfe algorithm \citep{frank1956algorithm} and is introduced in \cite{bredies2009generalized}. The idea is to find the minimum of the functionnal $G$ by solving a semilinearized version of $G$ and considering a suitable convex combination of the minimizers at each step. 
We show, under suitable assumptions, that this procedure is equivalent to Algorithm \ref{algorithm} which consists in iterating the fixed point for a certain sequence of weights $(\alpha_k)_{k \in \mathbb{N}}$. To describe the algorithm we define the following set: 
$$
\mathcal{C} \coloneqq \left\{ X \in L^1_+(\mathcal{F}^0_T), \, \bar{G} \leq G(X) \leq G(\mathds{1})\right\},
$$
where $\bar{G} \coloneqq \inf_{\xi \in L^1_+(\mathcal{F}_T^0)} G (\xi)$.

\begin{algorithm}[H]
    \caption{Frank-Wolfe} \label{algorithm}
    \begin{algorithmic}
    \STATE Let $(\alpha_k)_{k \in \mathbb{N}}$ be a sequence in $(0,1)$ and choose $\xi_{0} \in \mathcal C$.\;
    \FOR{$0\leq k < N$} 
    \STATE {
    1. Compute \begin{equation*}
        \eta_k = \frac{\exp\left(-\gamma^* \mathbb{E}[  V^{\xi_k}_T|\mathcal F^0_T]+ \rho \ln Z\right)}{\mathbb{E}\left[\exp\left(-\gamma^* \mathbb{E}[  V^{\xi_k}_T|\mathcal F^0_T] + \rho \ln Z\right)\right]},
    \end{equation*} \\
    2. Compute \begin{equation*}
        \xi_{k+1} = \alpha_k \eta_k + (1-\alpha_k) \xi_{k}. 
    \end{equation*}
    }
    \ENDFOR
    \RETURN {$(\eta_N, \xi_N)$.}
    \end{algorithmic}
\end{algorithm}

 \begin{theorem} \label{thm:convergence-result} 
Assume that the parameters are such that
$$
L(\mathds{1})<+\infty, \quad \mathbb{E} \left[ \int_0^T \frac{e_t^2}{\alpha_t} \mathbb{E} \left[ Z^\rho \mathcal{E}_{t,T}|\mathcal F_t \right]^2  \dd t\right]<+\infty.
$$
Let $p\geq 2$ be such that $\frac{2C}{p}\leq G(\mathds{1})- G(\bar \xi)$, where $C$ is the constant introduced in Lemma \ref{1iter.lm}. For all $k\in \mathbb{N}$, let $\alpha_k = \frac{2}{k+p}$  and let $(\xi_k)_{k\in \mathbb{N}}$ be the sequence defined by Algorithm \ref{algorithm}. Then,  
$$
G(\xi_k)-G(\bar \xi) \leq \frac{1}{k+p} \left(4C + \frac{K}{(k+p)} \right), \quad K = \left(G(\xi_0)-G(\bar \xi) - \frac{4C}{p} \right)^+.
$$
 \end{theorem}
 
\begin{remark}
The theorem shows that the algorithm converges when using step size $\alpha_k = \frac{2}{k+p}$ with $p$ large enough. It implies that $(\xi_k)$ is a minimizing sequence and hence, similarly to the proof of Lemma \ref{lemma:optG}, it follows that $(\xi_k)$ weakly converges to $\bar\xi$.  In practice, $G(\bar\xi)$ is not known but one can try increasing values of $p$ until convergence is obtained. 
\end{remark}

\subsection{Proof of Theorem \ref{main.thm}}
\label{sec:proof-of-main-thm}

To prove Theorem \ref{main.thm} we show that the equilibrium stochastic discount factor is uniquely defined as the solution of the convex minimization problem,
\begin{equation} \label{pb:optim-xi} \tag{$\Pb$}
\inf_{\xi \in  L^1_+(\mathcal F^0_T)} G(\xi).
\end{equation}
The proof is organized as follows: we first show that the mapping $G$ is proper, strictly convex and weakly lower semicontinuous (Lemma \ref{lemma:G}). Next we show that every element of the set $\mathcal C$ belongs to $\Xi$ (Lemma \ref{lemma:bounded-entropy}). Then we show that $G$ admits a unique minimizer $\bar \xi\in \mathcal C$ (Lemma \ref{lemma:optG}), and hence $\bar\xi$ belongs to $\Xi$.  This minimizer satisfies $\bar \xi >0$ a.s. (Lemma \ref{lemma:posG}). Finally we show that every minimizer of $G$ is a solution to the fixed point problem and every solution to the fixed point problem is a minimizer of $G$ (Lemma \ref{iff.lm}).

\begin{lemma} \label{lemma:G}
    The mapping $G$ is proper, strictly convex and weakly lower semicontinuous.
\end{lemma}

\begin{proof}
  \textit{Step 1: $G$ is proper and strictly convex.}
By assumption $G(\mathds{1})= L(\mathds{1}) < +\infty$. For all $\xi \in L^1_+(\mathcal{F}^0_T)$ such that $\mathbb{E} \left[\xi \right] = 1$, we have 
\begin{equation*}
    G(\xi) \geq H(\xi),
\end{equation*}
since $L(\xi) \geq 0$.  Using that $\inf_{x > 0} y^{\rho}h(x/y^{\rho}) \geq - y^{\rho}$, for all $y \in \mathbb{R}_+$, we have
\begin{equation*}
     Z^{\rho} h(\xi/Z^{\rho}) \geq -Z^\rho,
\end{equation*}
almost surely, for all $\xi \in L^1_+(\mathcal{F}^0_T)$. Then 
\begin{equation*} 
H(\xi) \geq \mathbb E\left[\inf_{\xi > 0} Z^{\rho} h(\xi/Z^{\rho})\right] \geq -\mathbb E\left[Z^\rho\right] \geq -\mathbb E\left[Z\right]^\rho= -1,
\end{equation*}
where the last inequality follows by Jensen's inequality, since $\rho<1$. Then $G$ is lower bounded and is thus proper.
The strict convexity of $G$ follows directly from the strict convexity of $H$, the convexity of $L$ as an upper bound of a family of linear mappings and the convexity of the mapping $\xi \mapsto \chi(\mathbb{E} \left[ \xi \right] - 1)$.

\noindent \textit{Step 2: $G$ is weakly lower semicontinuous.} We first prove that $G$ is strongly lower semicontinuous.
Since $G$ is lower bounded by step 1 we have that $\bar{G} = \inf_{\xi \in L^1_+(\mathcal{F}^0_T)} G(\xi)$ is finite.
Let $(\xi^n)_{n\in \mathbb{N}} \in \dom(G)$ be a sequence  stronlgy converging to $\xi$. Since strong $L^1$ convergence implies the convergence in probability, and convergence in probability implies almost sure convergence along a subsequence, up to choosing a subsequence we may and will assume that $\xi^n\to \xi$ almost surely. On the one hand, we have
\begin{equation*}
    Z^{\rho} h(\xi^n /Z^{\rho}) \underset{n \to +\infty}{\to} Z^{\rho} h(\xi /Z^{\rho}),
\end{equation*}
almost surely.  Since, as we have shown above, $Z^{\rho} h(\xi^n /Z^{\rho}) \geq - Z^{\rho}$, Fatou's Lemma yields
\begin{equation*}
\mathbb E\left[H(\xi^n) \right] \\\leq \liminf_{n \to +\infty} \mathbb E\left[H(\xi^n)\right].
\end{equation*}
Since $H$ is convex, it is also weakly lower semicontinuous by \cite[Corollary 3.9]{brezis}. 

On the other hand, consider the functional $L$ defined in \eqref{functionals.eq}. The first term in $L$ is (strongly) lower semicontinuous by Fatou's lemma. In the second term, by a conditional version of the same lemma
$$
\mathbb E[\xi\mathcal E_{t,T}] \leq \liminf_n \mathbb E[\xi^n\mathcal E_{t,T}]
$$
a.s.~for all $t\in[0,T]$. Since $\xi_t>0$ a.s.~for all $t\in[0,T]$ and $c^*$ is lower semicontinuous, by another application of Fatou's lemma we get the weak lower semicontinuity of the mapping $L$. 
Since the sum of lower semicontinuous functions is lower semicontinuous, it follows that $G$ is strongly lower semicontinuous. 
\end{proof}

\begin{lemma} \label{lemma:bounded-entropy}
    There is a constant $C>0$ such that for any $\xi \in \mathcal{C}$,
        \begin{align*} 
            &\mathbb{E} \left[ \xi \ln(\xi) \right]\leq C,  &&\mathbb{E} \left[ \xi \ln\frac{\xi}{Z} \right]\leq C,\\
            &\mathbb E[V\xi \mathcal E_T]\leq C,  && \mathbb E\left[\int_0^T \alpha_t c^*\left(\frac{e_t}{\alpha_t} \mathbb E[\xi \mathcal E_{t,T}|\mathcal F_t]\right) \dd t\right]\leq C .
        \end{align*}
\end{lemma}
\begin{proof}
     We start with the two first inequalities. 
     Since $\xi \in \mathcal{C}$, we have $ H(\xi) + L(\xi) \leq G(\mathds{1})$.  Using that $L(\xi) \geq 0$ and the definition of $H$, we have
    \begin{equation}  \label{eq:h-lnZ}
        \mathbb{E}[\xi \ln(\xi)]  \leq  \mathbb{E}\left[ \rho \xi \ln(Z) \right] +1 + G(\mathds{1}).
    \end{equation}
    By convex duality we also have that 
    \begin{equation} \label{eq:lnZ-lnxi}
        \xi \ln(Z) \leq Z+ \xi \ln \xi - \xi.
    \end{equation}
    Then combining \eqref{eq:h-lnZ} and \eqref{eq:lnZ-lnxi}, we have 
    $$\mathbb{E} \left[ \xi \ln(\xi) \right]\leq \frac{ 1 + G(\mathds{1})}{1-\rho}.$$
    Now since 
    $$
    H(\xi) = (1-\rho)\mathbb E[\xi \ln \xi - \xi] + \rho \mathbb E[\xi \ln \frac{\xi}{Z}] - \rho
    $$
    and $\mathbb E[\xi \ln \xi - \xi]\geq -1$, 
    we conclude that 
    $$
      \mathbb E[\xi \ln \frac{\xi}{Z}] \leq \frac{G(\mathds 1)+ 1}{\rho}
    $$
   Finally, as in the proof of Lemma \ref{lemma:G}, we have that $H(\xi) \geq - \mathbb{E}[Z^\rho]$. Then
        $$
        L(\xi) \leq G(\mathds{1}) + \mathbb E[Z^\rho] \leq G(\mathds{1}) + 1.
        $$
        Since both terms in $L$ are positive, the last two inequalities follow. 
\end{proof}

\begin{lemma}
    There is a unique solution $\bar{\xi}\in \mathcal C$ of the problem \eqref{pb:optim-xi}.
    \label{lemma:optG}
\end{lemma}
The proof follows the direct method of variation calculus.

\begin{proof}
    For all $\xi \in \mathcal{C}$, by Lemma \ref{lemma:bounded-entropy} there exists $C>0$ such that $\mathbb{E} \left[ \xi \ln(\xi) \right] \leq C$. Then by de la Vall{\'e}e Poussin theorem
    \cite[Theorem VI]{poussin1915integrale} the set is uniformly integrable. Now by the Dunford-Pettis \cite[Theorem 4.30]{brezis} the set $\mathcal{C}$ is weakly relatively compact and by the Eberlein-\v{S}mulian Theorem \cite[Section V.6.1, p.430]{dunford1988linear}, the set $\mathcal{C}$ is weakly sequentialy relatively compact. Now let $(\xi^n)_{n \in \mathbb{N}} \in \mathcal{C}$ be a minimizing sequence, that is to say such that $\lim_{n \to \infty} G(\xi^n) = \bar{G}$. By weak compactness of $\mathcal{C}$, there exists $\bar{\xi}\in L^1_+(\mathcal F^0_T)$ such that the minimizing sequence weakly converges to $\bar{\xi}$. Using that $G$ is weakly lower-semicontinuous by Lemma \ref{lemma:G}, yields that
    \begin{equation*}
        \bar{G} \leq G(\bar{\xi}) \leq \liminf_{n \to \infty} G(\xi^n) \leq \lim_{n \to \infty} G(\xi^n)  = \bar{G},
    \end{equation*}
    which shows that $\bar{\xi}$ is a minimizer of $G$. Finally the uniqueness follows by strict convexity of the criterion $G$.
\end{proof}

\begin{lemma} The minimizer $\bar \xi$ of $G$ satisfies $\xi >0$ a.s.\label{lemma:posG}
\end{lemma}
\begin{proof}
Assume by way of contradiction that there is a set $A\in \mathcal F^0_T$ with $p_A:=\mathbb P[A]>0$ and $\bar\xi = 0$ on $A$. For $\varepsilon \in (0,1)$, define 
$$
\xi_\varepsilon = (1-\varepsilon)\bar \xi + \frac{\varepsilon}{p_A} \mathds 1_{A}.  
$$
It is easy to check that as $\varepsilon \downarrow 0$, 
$$
G(\xi_\varepsilon)-G(\bar \xi) = \varepsilon \ln(\varepsilon) + O(\varepsilon). 
$$
Since the right-hand side is negative for $\varepsilon$ small enough, this contradicts the fact that $\bar \xi$ is a minimizer of $G$. 
\end{proof}

\begin{lemma} \label{lemma:deriv}
    {
    Let $\xi \in \mathcal{C}$. For any $\eta \in L^1(\mathcal{F}^0_T)$ such that $\mathbb{E}\left[ \eta \right] = 0$ and $|\eta| \leq \frac{1}{2} \xi$ almost surely, we have
    \begin{equation} \label{eq:deriv-G}
       \lim_{\lambda \to 0} \lambda^{-1} \left(G(\xi + \lambda \eta)- G(\xi) \right) = \mathbb{E}\left[ \eta (\delta H(\xi) + \delta L(\xi))\right],
    \end{equation}
    with $\delta L(\xi) \coloneqq  \mathbb{E}[V_T^\xi \vert \mathcal{F}_T^0]$ and $\delta H(\xi) \coloneqq  \ln(\xi/Z^{\rho})$.}
\end{lemma}

\begin{proof}
    For any $\lambda \in (0,1)$ we have
\begin{equation*}
\lambda^{-1} \left(G(\xi + \lambda \eta)-G(\xi) \right) = \lambda^{-1} \left((H+L)(\xi + \lambda \eta) - (H+L)(\xi) \right), 
\end{equation*}
since $\mathbb{E}\left[\xi + \lambda \eta \right] = \mathbb{E}\left[ \xi \right] = 1$.
We consider separately the functionals $H$ and $L$. For the functional $H$, we have
\begin{equation} \label{deriv}
    \lambda^{-1} (H(\xi + \lambda \eta) - H( \xi)) = \lambda^{-1} \mathbb{E} \left[Z^{\rho} \left(h((\xi + \lambda \eta)/Z^{\rho}) - h(\xi/Z^{\rho})\right)\right].
\end{equation}
Let $(x,y) \in \mathbb{R}_+ \times \mathbb{R}$, and assume that $\lambda |y| \leq x$. By the fundamental theorem of calculus we have
\begin{equation*}
    \lambda^{-1}(h(x + \lambda y) - h(y)) = y \int_{0}^{1} \ln (x +  t \lambda y)\, \dd t.
\end{equation*}
By monotonicity and concavity of the logarithm we have
\begin{equation*}
    \ln(x) + \ln(1-t)  \leq  \ln(x + \lambda t \eta) \leq \ln(x) + \ln(2),
\end{equation*}
for any $t\in [0,1]$. Then we have that
\begin{equation*}
    y \ln(x) + y \int_0^1 \ln(1-t) \dd t \leq \lambda^{-1}(h(x + \lambda y) - h(y)) \leq y (\ln(\xi) + \ln(2)).
\end{equation*}
Combining the above estimates and using both that $\mathbb{E}[\eta] = 0$ and $|\eta|\leq \frac{1}{2}\xi$ yields
\begin{equation*}
   \lambda^{-1}  \vert \mathbb{E} \left[Z^{\rho} h((\xi + \lambda \eta)/Z^{\rho}) - h(\xi/Z^{\rho})\right]\vert = \mathbb{E} \left[|\eta \ln(\xi/Z^{\rho})|\right] \leq \frac{1}{2} \vert H(\xi) \vert \leq C.
\end{equation*}
where the last inequality follows by Lemma \ref{lemma:bounded-entropy}.
We conclude that \eqref{deriv} admits an integrable bound and so by dominated convergence
$$
\lim_{\lambda\downarrow 0} \lambda^{-1}(H(\xi + \lambda \eta)-H(\xi)) = \mathbb E[\eta \ln (\xi)]. 
$$
We now consider the functional $L$. Using that $c^*$ is assumed to be twice differentiable yields
\begin{align*}
\lambda^{-1}(L(\xi + \lambda \eta)-L(\xi))  = \gamma^\star   \mathbb{E} \left[\eta \left( V\mathcal{E}_T + \int_0^T \alpha_t (c^*)'\left(\psi^\xi_t\right) \dd t \right) \right]  + \gamma^\star\frac{\lambda}{2}  \mathbb{E} \left[\int_0^T \alpha_t  (c^*)''\left(\psi^\xi_t\right)|\psi^{\eta}_t|^2 \dd t \right].
\end{align*}
Since $(c^*)''$ is bounded by assumption, it is clear that the last term converges to zero as $\lambda \downarrow 0$. Using that $\eta \in L^1(\mathcal{F}_T^0)$, the conclusion follows applying the law of iterated expectations and using that $V_T^\xi$ is solution to equation \eqref{eq:psi-V}.
\end{proof}

\begin{lemma}\label{iff.lm}
$\bar \xi$ is a minimizer of $G$ if and only if it is a solution to the fixed point problem \eqref{pb:fixed-point}.
\end{lemma}
\begin{proof}
    \textit{Step 1: The if part}. Let $\bar{\xi} \in \Xi$ be a solution to the fixed point problem \eqref{pb:fixed-point}. For any $\xi \in L^1_+(\mathcal F^0_T)$ with $\mathbb E[\xi]=1$, by convexity of $L$ and $H$, we have
\begin{equation*}
 G(\xi) - G(\bar{\xi})  \geq \mathbb E\left[(\xi - \bar{\xi})\left( \delta H(\bar{\xi}) + \delta L(\bar{\xi})\right)\right].
\end{equation*}
By Lemma \ref{lemma:deriv} and using that $\bar \xi$ is a solution to the fixed point equation \eqref{pb:fixed-point} yields that 
\begin{align*}
    \mathbb E\left[(\xi - \bar{\xi})\left( \delta H(\bar{\xi}) + \delta L(\bar{\xi})\right)\right] &= \mathbb E\left[(\xi - \bar{\xi})\left(  \ln \bar{\xi} - \rho \ln Z + \gamma^* \mathbb E[V^{\bar{\xi}}_T|\mathcal F^0_T] \right)\right] \\
    & = C \mathbb E\left[(\xi - \bar{\xi})\right],
\end{align*}
with $C \coloneqq \ln\mathbb E\left[\exp\left(\rho \ln Z - \gamma^* \mathbb E[V^{\bar{\xi}}_T|\mathcal F^0_T] \right)\right] > 0$.
Since $\mathbb E[\xi] = \mathbb E[\bar{\xi}] = 1$, we conclude that $G(\xi) - G(\bar{\xi})\geq 0$ and hence $\bar{\xi}$ is a minimizer of $G$. 

\noindent \textit{Step 2: The only if part}. Let $\bar \xi$ be the minimizer of $G$. 
Let $\eta \in L^1(\mathcal F^0_T)$ such that $|\eta|\leq \frac{1}{2}\bar\xi$ a.s.~and $\mathbb E[\eta]=0$. For any $\lambda\in (0,1)$, Lemma \ref{lemma:deriv} yields
\begin{align*} \nonumber
\lambda^{-1} \left(G(\bar\xi + \lambda \eta)-G(\bar \xi) \right) = \mathbb{E}\left[ \eta (\delta H(\xi) + \delta L(\xi))\right] \geq 0. \label{ineq}
\end{align*}
Since  the choice of $\eta \in L^1(\mathcal F^0_T)$ is arbitrary provided that $|\eta|\leq \frac{1}{2}\xi$ and $\mathbb E[\eta]=0$, and since $\bar \xi >0$  a.s.,  we conclude that necessarily 
$$
\ln \bar \xi -\rho\ln (Z) + \gamma^* \mathbb E[V^{\bar{\xi}}_T|\mathcal F^0_T]  = C,
$$
for some constant $C$. The value of this constant may be found from the condition $\mathbb E[\bar \xi]=1$, leading to the fixed point equations \eqref{pb:fixed-point}.
\end{proof}

\subsection{Proof of Theorem \ref{thm:convergence-result}}
\label{sec:proof-of-convergence}

The proof of Theorem \ref{thm:convergence-result} will be given after a series of preparatory lemmas. We start with the following simple representation that follows directly from the definition of $L$.
\begin{corollary} \label{corollary:reg-L}
    For all $\xi, \eta \in L^1_+(\mathcal{F}^0_T)$, we have that
    \begin{equation*}
        L(\eta) -  L(\xi)  =  \mathbb{E} \left[(\eta-\xi) \delta L(\xi) \right] + R(\xi,\eta),
    \end{equation*}
    with 
    $R(\xi,\eta) \coloneqq \frac{\gamma^\star}{2} \mathbb{E} \left[ \int_0^T \alpha_t (c^*)''\left(\psi^\xi_t\right) |\psi_t^{\eta} -\psi_t^{\xi}|^2 \dd t\right] \geq 0.$
\end{corollary}
For any $\xi \in \mathcal{C}$ we define the semilinearized criterion $h \colon \mathcal{C}\times  L^1(\mathcal{F}_T^0) \to \mathbb{R} \cup \{+\infty\}$, and the point $\eta \in L_+^1(\mathcal F^0_T)$ as follows
\begin{equation} \label{eq:eta-xi}
    \eta = \argmin_{\eta' \in L^1(\mathcal F^0_T)} h(\xi,\eta'), \quad
    h(\xi,\eta) \coloneqq \mathbb{E}\left[\eta \delta L(\xi) \right] + H(\eta) + \chi \left(\mathbb{E}\left[\eta \right] - 1 \right).
\end{equation}

\begin{lemma} \label{lemma:eta-argmin-h}
    We have \eqref{eq:eta-xi} holds if and only if,
    \begin{equation} \label{eq:eta-explicit}
        \eta = \frac{\exp\left(-\gamma^* \mathbb{E}[ V^\xi_T|\mathcal F^0_T]+ \rho \ln Z\right)}{\mathbb{E}\left[\exp\left(-\gamma^* \mathbb{E}[ V^\xi_T|\mathcal F^0_T] + \rho \ln Z\right)\right]}.
    \end{equation}
\end{lemma}

\begin{proof} \textit{Step 1: The only if part.} Let $\lambda > 0$ and $\eta' \in L^1(\mathcal{F}_T^0)$ be such that $\mathbb{E}[\eta'] = 0$. Using that $\eta$ satisfies \eqref{eq:eta-xi}, we have that
    \begin{equation*}
        h(\xi,\eta + \lambda \eta') - h(\xi,\eta) =  \mathbb{E}\left[  \eta'\left( \delta L(\xi) + \delta H(\eta)  \right) \right] \geq 0.
    \end{equation*}
    Since the choice of $\eta$ is arbitrary, the conclusion follows as in step 2 in the proof of Lemma \ref{iff.lm}.

    \noindent \textit{Step 2: The if part.} For any $\eta' \in L_+^1(\mathcal{F}^0_T)$, by convexity of $h(\xi,\cdot)$, we have
    \begin{equation*}
        h(\xi,\eta') - h(\xi,\eta) \geq \mathbb{E}\left[ (\eta'-\eta)\left( \delta L(\xi) + \delta H(\eta)  \right) \right].
    \end{equation*}
    Using that $\eta$ is given by \eqref{eq:eta-explicit} yields that 
    \begin{align*} 
        h(\xi,\eta') - h(\xi,\eta) & \geq   \mathbb{E}\left[(\eta'-\eta) \left(\ln \bar{\eta} - \rho \ln Z + \gamma^* \mathbb E[V^{\eta}_T|\mathcal F^0_T] \right)\right] \\
        &\geq C \mathbb{E}\left[ (\eta'-\eta) \right] = 0,
    \end{align*}
    with $C \coloneqq \mathbb E\left[\exp\left( \delta L(\xi) + \delta H(\eta) \right)\right] > 0$ and where the last inequality holds since $\mathbb E[\eta'] = \mathbb E[\eta] = 1$. Then $\eta$ is a minimizer of $h(\xi,\cdot)$. 
\end{proof}

We define the primal gap $\varepsilon \colon \mathcal{C} \to \mathbb{R}$ and the primal gap certificate $g \colon \mathcal{C} \to \mathbb{R}$,
\begin{equation} \label{def:gaps}
    \varepsilon(\xi) = G(\xi) - G(\bar{\xi}), \quad 
    g(\xi) = h(\xi,\xi) - \inf_{\eta \in L^1(\mathcal{F}^0_T)} h(\xi,\eta).
\end{equation}

\begin{lemma} \label{lemma-gap-vareps}
    For any $\xi \in \mathcal{C}$ we have that $g(\xi) \geq 0$ and $g(\xi) \geq \varepsilon(\xi)$.
\end{lemma}
\begin{proof}
    It is clear that the primal gap certificate $g$ is non-negative for all $\xi \in \mathcal{C}$. 
    By definition of $h$ we have that
    \begin{align*}
        h(\xi,\eta) - h(\xi,\xi) =   \mathbb{E}\left[(\eta - \xi) \delta L(\xi) \right] + H(\eta) + \chi \left(\mathbb{E}\left[\eta \right]- 1 \right)
           - \chi \left(\mathbb{E}\left[\xi \right] - 1 \right).
    \end{align*}
    By Corollary \ref{corollary:reg-L} we have that
    \begin{equation*}
        L(\eta) - L(\xi)  \geq \mathbb{E}\left[(\eta - \xi) \delta L(\xi) \right],
    \end{equation*}
    since $(c^*)''$ is non-negative, which yields that
    \begin{equation*} 
        h(\xi,\eta) - h(\xi,\xi) \leq G(\eta) - G(\xi).
    \end{equation*}
     Taking the infimum with respect to $\eta \in L^1(\mathcal{F}^0_T)$ both sides of the latter inequality concludes the proof. 
\end{proof}

\begin{lemma}\label{1iter.lm}Let the parameters be such that $L(\mathds{1})<+\infty$ and moreover 
    $$
    \mathbb{E} \left[ \int_0^T \frac{c_t^2}{\alpha_t} \mathbb{E} \left[ Z^\rho \mathcal{E}_{t,T}|\mathcal F_t \right]^2  \dd t\right]<+\infty.
    $$
    For any $\alpha \in [0,1]$, there exists a constant $C>0$, which does not depend on $\xi$, such that 
    $$
    G(\xi^\alpha) - G(\bar \xi) \leq (1-\alpha)(G(\xi) - G(\bar \xi)) + \alpha^2 C. 
    $$
    \end{lemma}
    
    \begin{proof}
    By convexity of the functional $H$, we have:
    \begin{align} \label{eq:convexity-H}
         H(\xi^\alpha) -  H(\xi) & \leq \alpha \left(H(\eta)  - H(\xi)  \right).
    \end{align}
    Combining \eqref{eq:convexity-H} with Corollary \ref{corollary:reg-L} yields 
    \begin{align*}
        G(\xi^\alpha) -  G(\xi) &\leq \alpha(H(\eta) - H(\xi)) +\alpha\mathbb E[(\eta-\xi)\delta L(\xi)] + \alpha^2 R(\xi,\eta).
    \end{align*}   
    By Lemma \ref{lemma-gap-vareps} it follows that
    \begin{align*}
        G(\xi^\alpha) -  G(\xi) &\leq \alpha(h(\xi,\eta) - h(\xi,\xi)) + \alpha^2 R(\xi,\eta)\\
        & \leq - \alpha g(\xi) + \alpha^2 R(\xi,\eta)\\
        &\leq \alpha(G(\bar \xi)-G(\xi)) + \alpha^2 R(\xi,\eta).
    \end{align*}   
    where in the last inequality we used the convexity of $L$. Rearranging the terms, we then get:
    $$
    G(\xi^\alpha) - G(\bar \xi) \leq (1-\alpha)(G(\xi) - G(\bar \xi)) + \alpha^2 R(\xi,\eta). 
    $$
    We now turn to the analysis of the residual term. Using that $(c^*)''$ is bounded, the Young inequality yields
    \begin{align*}
    R(\xi,\eta) & \leq {\gamma^\star} C \left( \mathbb{E} \left[ \int_0^T \frac{c_t^2}{\alpha_t} \mathbb{E} \left[ \xi \mathcal{E}_{t,T}|\mathcal F_t \right]^2  \dd t\right]+ \mathbb{E} \left[ \int_0^T \frac{c_t^2}{\alpha_t} \mathbb{E} \left[ \eta \mathcal{E}_{t,T}|\mathcal F_t \right]^2  \dd t\right]\right) \\
    & \leq C\left(1 + {\gamma^\star} \mathbb{E} \left[ \int_0^T \frac{c_t^2}{\alpha_t} \mathbb{E} \left[ \eta \mathcal{E}_{t,T}|\mathcal F_t \right]^2  \dd t\right]\right),
    \end{align*}
    where the second line follows by Lemma \ref{lemma:bounded-entropy}. Now by Lemma \ref{lemma:eta-argmin-h}, we have that $\eta$ is given by \eqref{eq:eta-explicit}.
    To bound the second term, we use the explicit form of $\eta$. 
    By definition of $\delta L$ and the Jensen's inequality, 
    \begin{align*}
    \mathbb{E}\left[\exp\left(-\gamma^* \mathbb{E}[ V^\xi_T|\mathcal F^0_T] + \rho \ln Z\right)\right] & =  \mathbb{E}\left[\exp\left(- \delta L(\xi) + \rho \ln Z\right)\right]  \\
    & \geq \exp\left( - \mathbb{E}\left[ \delta L(\xi) \right] + \rho \mathbb E[\ln Z]\right).
    \end{align*}
    By Lemma \ref{lemma:bounded-entropy} we have that $\mathbb{E}\left[ \delta L(\xi) \right] \leq C$. Since  $\rho \mathbb E[\ln Z]$ is bounded by assumption, there exists a constant $c>0$, which does not depend on $\xi$ such that,
    $$
    \mathbb{E}\left[\exp\left(-\gamma^* \mathbb{E}[ V^\xi_T|\mathcal F^0_T] + \rho \ln Z\right)\right] \geq c,
    $$
    and therefore $\eta \leq c^{-1} Z^\rho$. Substituting this into the expression for the second part of the residual term $R(\xi,\eta)$, we get:
    $$
    \mathbb{E} \left[ \int_0^T \frac{c_t^2}{\alpha_t} \mathbb{E} \left[ \eta \mathcal{E}_{t,T}|\mathcal F_t \right]^2  \dd t\right] \leq c^{-2}\mathbb{E} \left[ \int_0^T \frac{c_t^2}{\alpha_t} \mathbb{E} \left[ Z^\rho \mathcal{E}_{t,T}|\mathcal F_t \right]^2  \dd t\right],
    $$
    which is finite by assumption of the lemma. 
    \end{proof}
    
\begin{proof}[Proof of Theorem \ref{thm:convergence-result}.]
Let $k \in \mathbb{N}$ and let $\varepsilon_k = G(\xi_k)-G(\bar \xi)$. Assume that $\xi_k \in \mathcal C$, then by Lemma \ref{1iter.lm}, 
\begin{align*}
\varepsilon_{k+1} &\leq \frac{k+p-2}{k+p}\varepsilon_k + \frac{4C}{(k+p)^2} \\ &\leq G(\mathds{1})-G(\bar \xi) - \frac{2}{k+p}(G(\mathds{1})-G(\bar \xi) ) + \frac{4C}{(k+p)^2}\\
&\leq G(\mathds{1})-G(\bar \xi) + \frac{4C - 2(k+p)(G(\mathds{1})-G(\bar \xi))}{(k+p)^2}\leq G(\mathds{1}) - G(\bar \xi),
\end{align*}
by definition of $p$. Therefore, $\xi_{k+1} \in \mathcal C$, and by induction $\xi_k \in \mathcal C$ for all $k \in \mathbb{N}$. 
It follows that
$$
\varepsilon_{k+1} - \frac{4C}{k+p+1} \leq \frac{k+p-2}{k+p} \left(\varepsilon_{k} - \frac{4C}{k+p} \right)
$$
and therefore
$$
\varepsilon_{k} - \frac{4C}{k+p} \leq \left(\varepsilon_{0} - \frac{4C}{p} \right) \prod_{j=0}^{k-1} \frac{j+p-2}{j+p}. 
$$
The product can be estimated as follows:
\begin{align*}
\ln \left(\prod_{j=0}^{k-1} \frac{j+p-2}{j+p}\right) = \sum_{j=0}^{k-1} \ln\left(1-\frac{2}{j+p}\right) 
 \leq- 2 \sum_{j=0}^{k-1} \frac{1}{j+p}\leq -2\int_0^k \frac{\dd x}{x+p} = -2\ln(k+p).
\end{align*}
Gathering all estimates, we finally obtain:
$$
\varepsilon_k \leq \frac{4C}{k+p} + \left(\varepsilon_{0} - \frac{4C}{p} \right)^+\frac{1}{(k+p)^2}. 
$$
 \end{proof}

\section{Examples and illustrations}
\label{examples.sec}
Let $B^0$ be a $2$-dimensional Brownian motions with components $(B^{0,1},B^{0,2})$. 
We assume that the process $(e_t)$ describing the emission efficacy of production is constant in time for each firm, $e_t = E$,  and that the process $(\alpha_t)$ describing the emission penalty is the same for all firms (and $\mathbb F^0$-measurable), and that the cost function is quadratic: $c(x) = \frac{x^2}{2}$. Although the emission efficacy is constant in time, the emission schedules of firms can still vary stochastically, in response to changes in the emission penalty $\alpha_t$ and the financial value of emissions, as explained in Remark \ref{emission.rem}.

The differences between firms appear only due to different firm values and different values of the emission efficacy $E$. 
We model the differences in initial firm values and emission efficacies by introducing $\mathcal F_0$-measurable random variables $V$ and $E$, which describe the distribution of these quantities. For green, carbon-efficient firms, $E$ is large, while for brown firms, $E$ is small. 
The emission penalty $(\alpha_t)$ is a stochastic process defined by
$$
\alpha_t = e^{\gamma B^{0,2}_t - \frac{\gamma^2}{2}t},
$$
where $\gamma \in \mathbb R$ is a constant, which measures the uncertainty associated to future emission penalty, in other words $\gamma$ is a climate transition risk parameter. To simplify notation and without loss of generality (by normalizing the emission efficacy $E$) we have assumed here that $\alpha_0 = 1$. 
Thus, under the real-world probability the emission penalty is a martingale, and has constant expectation. This process models the evolution of the environment in which all firms operate including carbon taxes, consumer preferences, etc.

The density of the change of measure of the green investor is defined by
$$
Z = e^{\lambda B^{0,2}_T - \lambda^2 T/2},
$$
where $\lambda \in \mathbb R$ is a constant. Under the probability measure of the green investor, $B^{0,g}_t = B^{0,2}_t - \lambda t$ is a Brownian motion and 
$$
\alpha_t = e^{\gamma B^{0,g}_t + \lambda \gamma t - \frac{\gamma^2}{2}t},
$$
has increasing expectation (assuming $\lambda$ and $\gamma$ are positive). Thus, green investors are concerned that costs and penalties associated to carbon emissions will grow. The parameter $\lambda$ modulates the strength of this effect, in other words, it can be interpreted as the environmental concern of the green investors. To give an example, assume that $\gamma = 0.3$. Then, the variance of $\alpha_T$ for $T=5$ years equals $\approx 57\%$. If $\lambda\gamma = 0.1$, then, under the green investor probability, the emission penalty will grow by $65\%$ over 5 years. 

The parameters of the firm value dynamics $\sigma$, $\sigma^0$ and $\mu$ are assumed to be constant, and we let 
$$
\overline{\mathcal E}_{t,T} \coloneqq \mathbb{E} \left[\mathcal E_{t,T} \vert \mathcal{F}_T^0 \right] =  e^{ \sigma^0 (B^{0,1}_T-B^{0,1}_t) + (\mu  - (\sigma^0)^2/2)(T-t)}.
$$
Since $\sigma$ is assumed to be constant, it does not appears in the expression of $(\overline{\mathcal E}_{t,T})_{t\in [0,T]}$.
The fixed point equation then writes
\begin{align*}
\xi & = \frac{\exp\left(-\gamma^*  V^\xi_T+ \rho (\lambda B^{0,2}_T - \lambda^2 T/2 )\right)}{\mathbb{E}\left[\exp\left(-\gamma^*  V^\xi_T + \rho (\lambda B^{0,2}_T - \lambda^2 T/2 )\right)\right]}, \\
V^\xi_T  & = \overline{V} \,\overline{\mathcal{E}}_T +  \int_0^T {{\overline{C^2}} e^{-\gamma B^{0,2}_t +|\gamma|^2 t/2}} \overline{\mathcal{E}}_{t,T} \, \mathbb{E}[\xi \overline{\mathcal{E}}_{t,T}|\mathcal F^0_t] \dd t,
\end{align*}
where we denote $\overline{E^2} = \mathbb E[E^2]$ and $\overline V = \mathbb E[V]$. 


\paragraph{Detailed implementation of Algorithm \ref{algorithm}}
We start by discretizing the integral: 
$$
\widehat V^\xi_T  = \overline{V} \,\overline{\mathcal{E}}_T +  h\sum_{k=0}^n w_k {\overline{E^2} e^{-\gamma B^{0,2}_{t_k} + |\gamma|^2 t_k/2}} \overline{\mathcal{E}}_{t_k,T} \, \mathbb{E}[\xi \overline{\mathcal{E}}_{t_k,T}|\widehat{\mathcal  F}_{k}],
$$
where $h = \frac{T}{n}$, $t_i = ih$, $w_0 = w_n = \frac{1}{2}$, $w_k = 1$ for $k=1,\dots,n-1$ and 
\begin{equation*}
    \widehat{\mathcal F}_{k} = \sigma\left(\varepsilon^1_j, \varepsilon^2_j, j\leq k \right), \quad \varepsilon^i_j = \frac{1}{\sqrt{h}}\left(B^{0,i}_{t_j} - B^{0,i}_{t_{j-1}}\right), \quad \forall i \in \{1,2\}.
\end{equation*}
This leads to the following algorithm:
\begin{enumerate}
\item Fix $p$ sufficiently large (different values may be tested until convergence is obtained). 
\item Simulate independent standard normal random increments: $(\varepsilon^{1,i}_j,\varepsilon^{2,i}_j)$ for $j=1,\dots,n$, $i=1,\dots,N$.  
\item Let $\xi^{(0),i}=1$, for $i=1,\dots,N$. 
\item For $q=0,\dots, N_{iter}$:
\begin{itemize}
\item For $k=0,\dots,n$, compute 
$$
v^{(q)}_k = \arg\inf_{v \in \mathcal C_{2k}} \sum_{i=1}^N \left(v(\varepsilon^{1,i}_1,\varepsilon^{2,i}_1,\dots,\varepsilon^{1,i}_k,\varepsilon^{2,i}_k) - \overline{\mathcal E}^i_{t_k,T} \xi^{(q),i}\right)^2.
$$
\item Compute
\begin{align*}
\eta^{(q+1),i} = \frac{\exp\left(-\gamma^*  \widehat V^{(q+1),i}_T+ \rho (\lambda \sqrt{h}\sum_{j=1}^n \varepsilon^{2,i}_j - \lambda^2 T/2 )\right)}{\frac{1}{N}\sum_{i=1}^N\exp\left(-\gamma^*  \widehat V^{(q+1),i}_T + \rho (\lambda \sqrt{h}\sum_{j=1}^n \varepsilon^{2,i}_j - \lambda^2 T/2 )\right)}.
\end{align*}
where the terminal conditional expectation is given by
\begin{align*}
    \widehat V^{(q+1),i}_T  =  \overline{V} \,\overline{\mathcal{E}}^i_{0,T} +  h\sum_{k=0}^n w_k {\overline{E^2} e^{-\gamma \sqrt{h}\sum_{j=1}^k \varepsilon^{2,i}_j + |\gamma|^2 t_k/2}} \overline{\mathcal{E}}^i_{t_k,T} \, v^{(q)}_k. 
\end{align*}
\item Actualize $\xi^{(q+1),i} = \alpha_q \eta^{(q),i} + (1-\alpha_q)\xi^{(q),i},\quad \alpha_q = \frac{2}{p+q}$.
\end{itemize}
\end{enumerate}
The computation of $v^{(q)}_k$, $k=1,\dots,n-1$ is implemented with a deep neural network using Keras-Tensorflow framework \citep{tensorflow2015-whitepaper}. Here,  $\mathcal C_{2k}$ denotes the set of continuous functions from $\mathbb R^{2k}$ to $\mathbb R$. 

\begin{remark}
    Since we show the convergence of Algorithm \ref{algorithm} in the continuous setting, we follow a ``first optimize then discretize" approach. That is to say, we derive an algorithm to compute the solution in the infinite dimensional setting and then discretize the optimality conditions.  A complete numerical analysis of the proposed algorithm should contain a detailed error analysis (which is out of the scope of the article) and show that the discrete solution tends to the continuous solution. Let us mention that approximation results via neural networks of conditional expectation can be found in \cite{cheridito2021computation}.
\end{remark}
\paragraph{Outputs of the algorithm} In addition to the quantities listed above, the algorithm allows to compute the following economically relevant quantities:
\begin{itemize}
\item The total average emissions, given by
$$
\overline \Psi_T = \int_0^T \mathbb E[\psi_t|\mathcal F^0_T] \dd t = \int_0^T {\overline E e^{-\gamma B^{0,2}_t + \gamma^2 t/2}} \mathbb E[\xi \overline {\mathcal E}_{t,T}|\mathcal F^0_t] \dd t. 
$$
The discretized version after $q$ iterations writes
\begin{align*}
\overline \Psi^{q+1}_T &=  h\sum_{k=0}^n w_k {\overline{E} e^{-\gamma \sqrt{h}\sum_{j=1}^k \varepsilon^{2,i}_j + |\gamma|^2 t_k/2}} v^{(q)}_k(\varepsilon^{1,i}_1,\varepsilon^{2,i}_1,\dots,\varepsilon^{1,i}_k,\varepsilon^{2,i}_k).
\end{align*}
\item The expected emissions of the representative company at date $t$, given by
$$
\mathbb E[\psi_t|\mathcal F_0] = E\mathbb E\left[\xi e^{-\gamma B^{0,2}_t + \frac{\gamma^2 t}{2}} \overline{\mathcal E}_{t,T}\right].
$$
The discretized version of this quantity is:
$$
\frac{E}{N}\sum_{i=1}^N \xi^{(q),i}  e^{-\gamma \sqrt{h} \sum_{j=1}^k \varepsilon^{2,i}_j + |\gamma|^2 t_k/2} \overline{\mathcal E}^i_{t_k,T}.
$$

\item The initial stock price of the representative company, given by
\begin{align}
S^\xi_0 &= V \mathbb E[\xi \mathcal E_T|\mathcal F_0] + \mathbb E\left[\int_0^T \frac{E^2}{\alpha_s} \mathbb E^2[\xi \mathcal E_{s,T}|\mathcal F_s]ds\Big|\mathcal F_0\right]\notag\\
& = V \mathbb E[\xi \overline{\mathcal E}_T] + E^2 \mathbb E\left[\int_0^T {e^{-\gamma B^{0,2}_t + \gamma^2 t/2}}\mathbb E^2[\xi \overline{\mathcal E}_{t,T}|\mathcal F_t] \dd t\right].\label{price.eq}
\end{align}
The discretized version of this formula writes:
$$
V v^{q}_0 + \frac{E^2 h}{N} \sum_{i=1}^N\sum_{k=0}^n w_k e^{-\gamma \sqrt{h}\sum_{j=1}^k \varepsilon^{2,i}_j + |\gamma|^2 t_k/2}\left(v^{(q)}_k(\varepsilon^{1,i}_1,\varepsilon^{2,i}_1,\dots,\varepsilon^{1,i}_k,\varepsilon^{2,i}_k)\right)^2.
$$
The first term in the above formula corresponds to the fraction of the market price that is due to the initial value of the firm, and the second term corresponds to the fraction of the price that is due to the value created by the firm through emissions. The second term is proportional to the squared emission efficacy of production; the higher the carbon intensity the lower will be the stock price and the higher the cost of capital. This effect is present in the market even if there are no green investors, because of the emission penalty, which is always present in the model. However, this effect will be stronger in the presence of green investors, or if their environmental concern is higher. 
\end{itemize}

We implement the algorithm to illustrate the impact of the climate risk and of the fraction and the environmental concern of green investors on the decarbonization dynamics and the cost of capital of firms with varying environmental performance. The parameters $\rho$ (proportion of green investors), $\gamma$ (volatility of emissions penalty, a proxy for climate risk) and $\lambda$ (proportional to environmental stringency of green investors) are changing throughout the tests and are described below. The other parameters are kept constant and given in the Table \ref{params.tab}. To fix the units, we take an average initial firm value of 1 billion US dollars, and an average emission efficacy of production to be 0.7 US dollars per kg of CO2. 

\begin{table}[]
    \centering
    \begin{tabular}{llp{0.5\textwidth}}
\textbf{Variable} & \textbf{Value} & \textbf{Description}\\
$T$ & 5 & Time horizon, years\\
$\gamma^*$ & $0.5$ & Risk aversion parameter \\
$\sigma_0$ & $10\%$ & Volatility of the common noise part of firm value dynamics\\
$\mu$ & $5\%$ & Drift of the firm value dynamics \\
$\overline V$ & 1 & Average initial firm value, \$billion \\
$\overline{E^2}$ & 1 & Average squared emission efficacy of production \\
 $\overline{E}$ & 0.7 & Average emission efficacy of production, \$/kgCO2 \\
$n$ & 20 & Number of discretization steps \\
$N$ & 50\,000 & Number of sample trajectories \\
$p$ & 2 & Weight in the fixed-point algorithm\\
\end{tabular}
    \caption{Table of parameters: value and description.}
    \label{params.tab}
\end{table}

\paragraph{Convergence} We first illustrate the convergence of the algorithm. Figure \ref{convergence.fig} plots the convergence of the distribution of $\xi^q$ and of the total average emissions $\overline \Psi^{q}_T$ as $q$ increases from $0$ to $9$. Here we took $\gamma = 0.3$ and $\lambda = 0$ (no green investors). The curves have been smoothed with a Gaussian kernel density estimator. We see that convergence is obtained starting from 5-6 iterations of the algorithm; in the numerical tests below we perform 10 iterations. 

\begin{figure}[h!]
\centerline{\includegraphics[width=\textwidth]{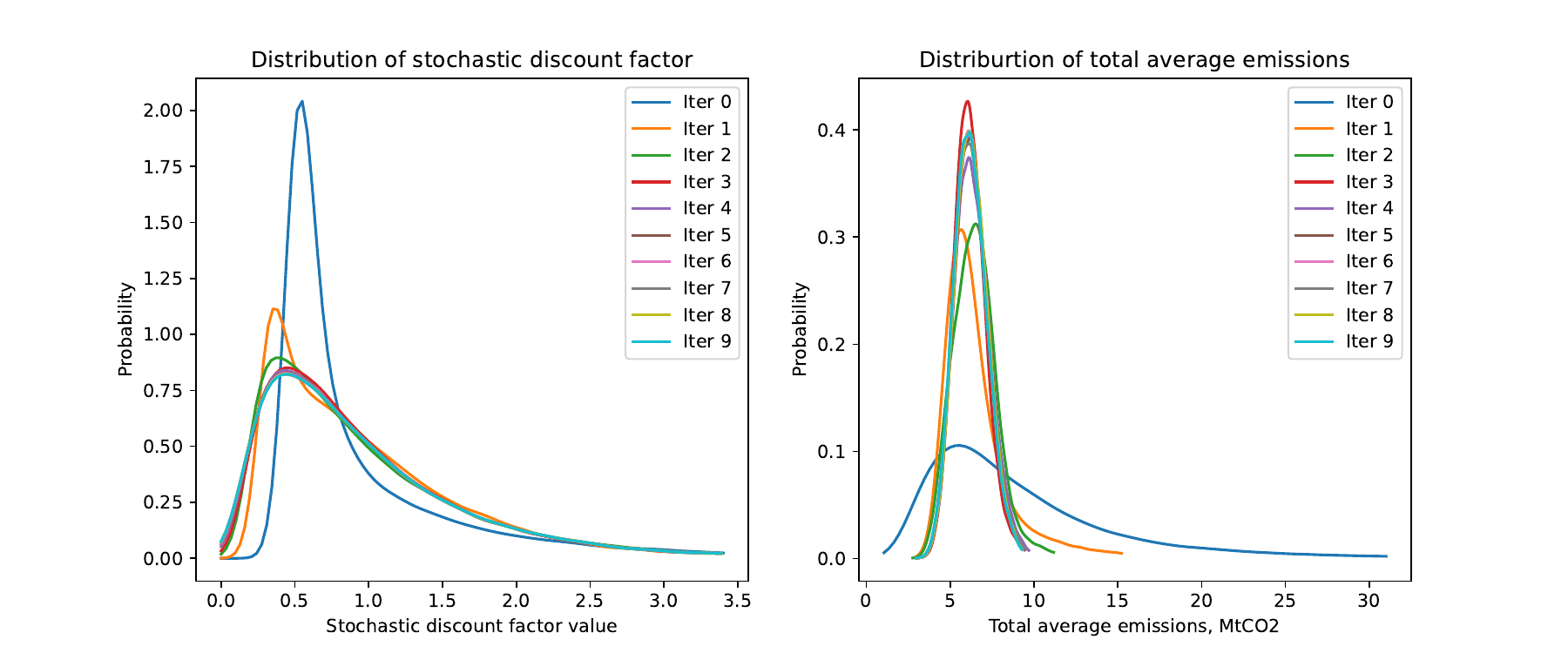}}
\caption{Convergence of the distribution of stochastic discount factor $\xi$ (left graph) and of the total average emissions. } 
\label{convergence.fig}
\end{figure}

\paragraph{Impact of climate risk on the decarbonization dynamics} In the first series of tests, we assume that there are no green investors ($\lambda=0$) and study the impact of climate risk on the decarbonization dynamics by changing the parameter $\gamma$ (volatility of the emission penalty).  Figure \ref{impact_gamma.fig} shows the distribution of total average emissions (left graph) and expected emissions of the representative company (with carbon intensity of production equal to 1) per unit of time (right graph), for different values of the volatility of emission penalty $\gamma$.  We see that as the intensity of climate risk/uncertainty grows, every company and the entire market increase their carbon emissions. This phenomenon has already been observed in \cite{de2022climate}: lax or uncertain climate policies do not encourage companies to reduce emissions. 

\paragraph{Impact of green investors on the decarbonization dynamics}
We now illustrate the impact of the proportion and the environmental concern of the green investors on the total average emissions of all companies. 
Figure \ref{impact_lambda.fig} shows the distribution of total average emissions (left graph) and expected emissions of the representative company  per unit of time (right graph), for different values of the environmental concern of green investors $\lambda$, with $\gamma = 0.3$ and $\rho = 0.5$. We see that as the environmental concern grows, the distribution of total average emissions shifts to the left and narrows down, while the expected emissions decrease for all dates. Figure \ref{impact_rho.fig} illustrates the impact of the proportion of green investors $\rho$, for $\gamma = 0.3$ and $\lambda = 0.4$, and we see that increasing the proportion of green investors leads to a considerable reduction of carbon emissions. 

This negative effect of climate policy uncertainty on the emission rates is thus partially reversed in the presence of environmentally concerned investors, whose impact on the cost of capital spurs companies to reduce emissions. However, as seen from Figure \ref{impact_rho.fig}, where for a relatively large value of volatility of emission penalty $\gamma$, the emission rates are increasing even for $\rho = 0.75$, if future climate policies are uncertain, even a large fraction of green-minded investors is unable to bring down the emission curve:  clear and predictable climate policies are an essential ingredient to allow green investors to decarbonize the economy. 

\begin{figure}[h!]
\centerline{\includegraphics[width=\textwidth]{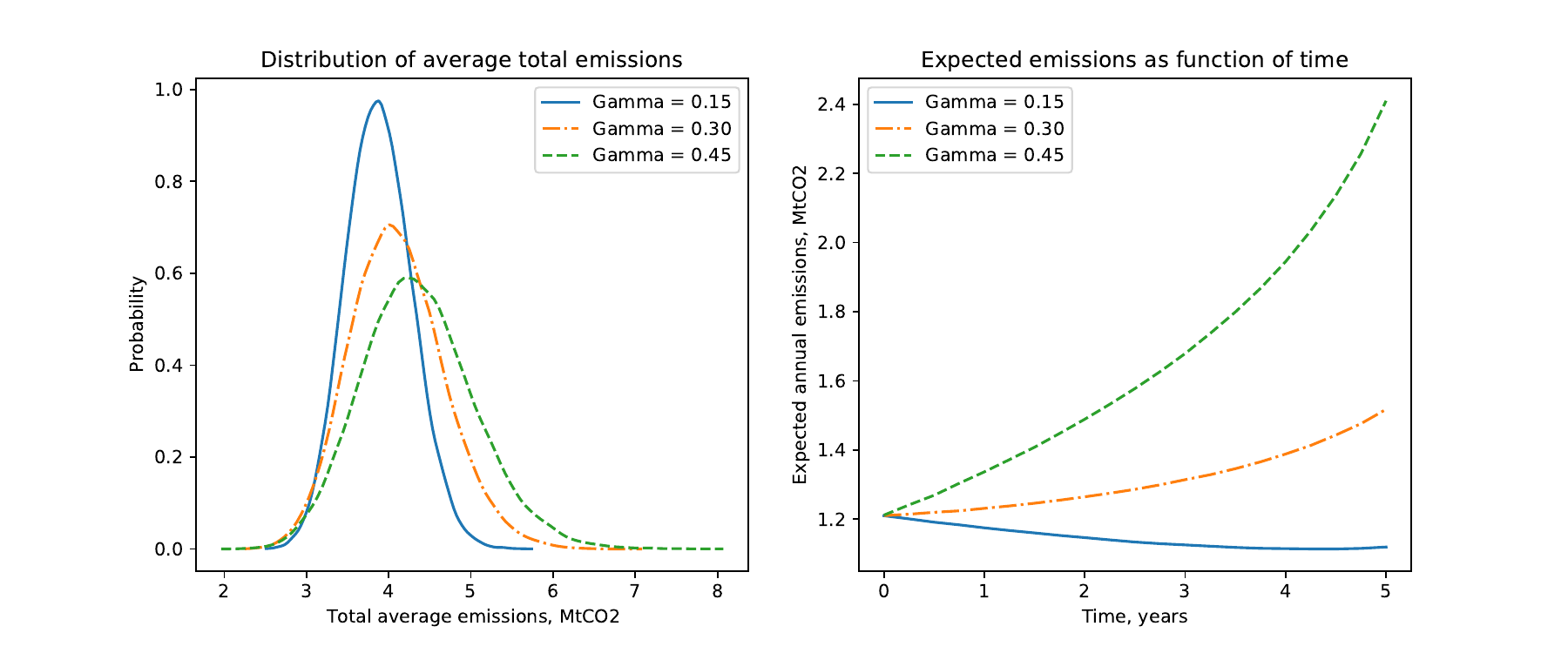}}
\caption{Distribution of total average emissions (left) and expected emissions of the representative company per unit of time (right), for different values of the volatility of emission penalty $\gamma$. } 
\label{impact_gamma.fig}
\end{figure}
\begin{figure}[h!]
\centerline{\includegraphics[width=\textwidth]{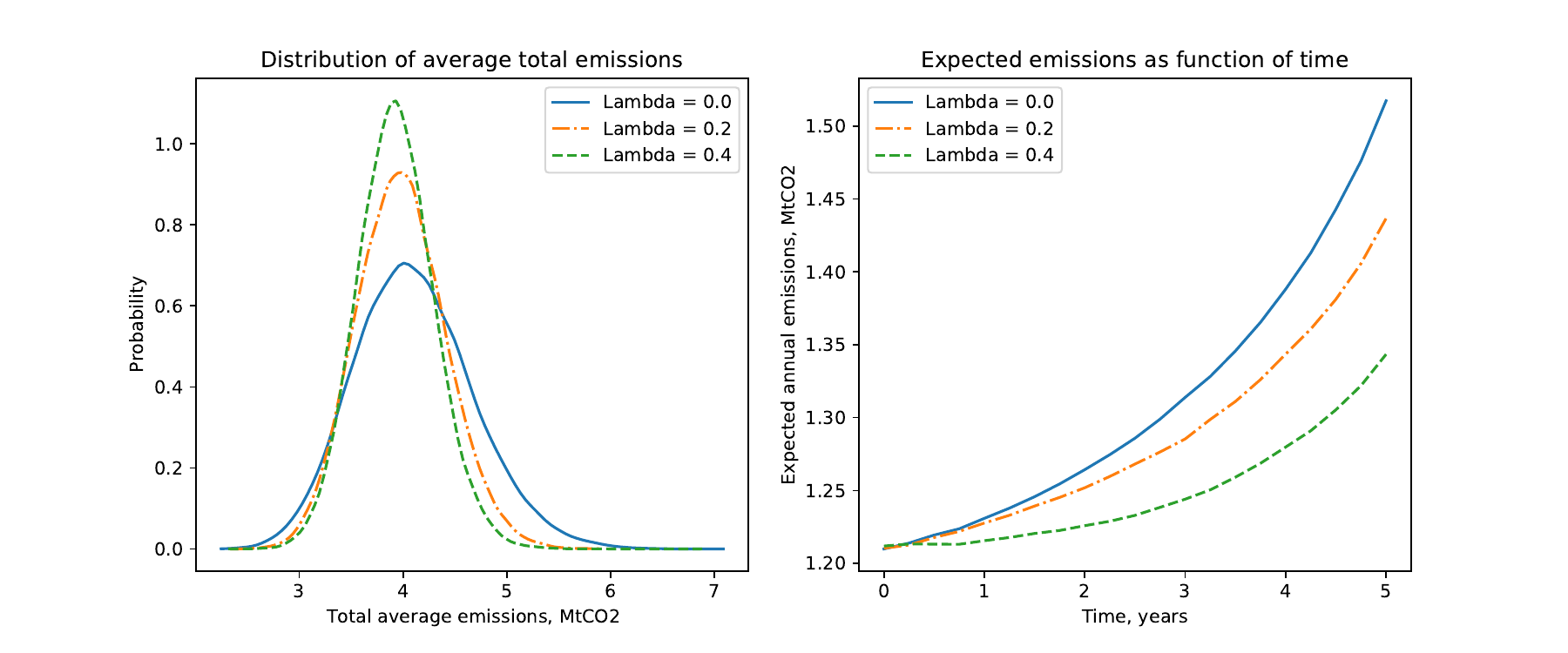}}
\caption{Distribution of total average emissions (left) and expected emissions of the representative company per unit of time (right), for different values of the environmental concern of green investors $\lambda$. } 
\label{impact_lambda.fig}
\end{figure}
\begin{figure}[h!]
\centerline{\includegraphics[width=\textwidth]{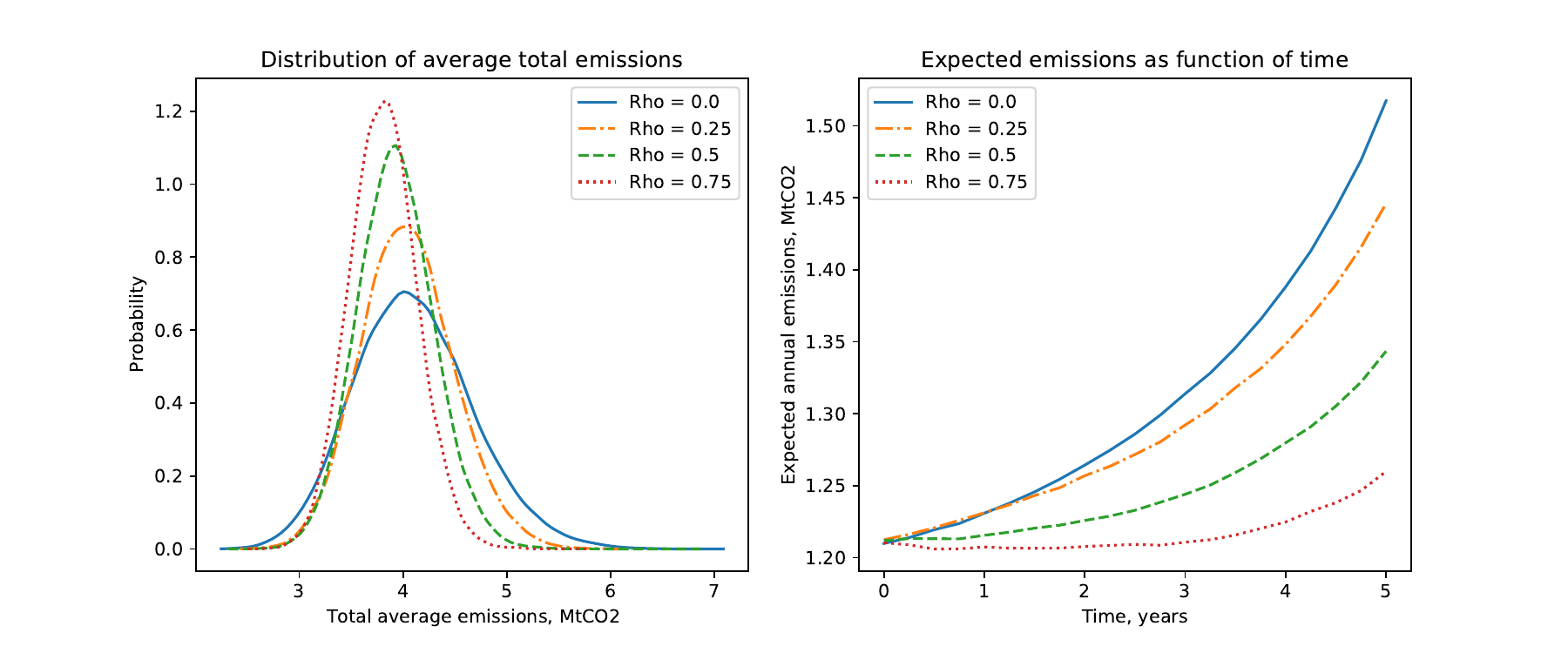}}
\caption{Distribution of total average emissions (left) and expected emissions of the representative company per unit of time (right), for different values of the proportion of green investors $\rho$. } 
\label{impact_rho.fig}
\end{figure}

\paragraph{Impact of the green investors on share prices}
We now discuss the impact of the climate risk, the fraction and the environmental concern of green investors on the share price of the representative company. Table \ref{price.tab} shows the impact of various parameters on the two components of the price formula \eqref{price.eq}: $P_1$ is the coefficient multiplied by the firm value $V$, and $P_2$ is the coefficient multiplied by the squared emission efficacy $E^2$. We see that the first component (sensitivity to firm value) is quite stable. Indeed, the quantity $\mathbb E[\xi \overline{\mathcal E}_T]$ only weakly depends on $\xi$, because $\mathbb E[\xi] = 1$ and the volatility $\sigma_0$ is not very high in our examples. On the other hand, the sensitivity to emission efficacy $P_2$ changes considerably between different tests. First, this sensitivity is always positive: all things being equal, carbon-efficient companies will enjoy higher stock prices than carbon-intensive companies. Second, this sensitivity increases when there is more climate risk / uncertainty (i.e., when $\gamma$ is high). This happens because, when climate uncertainty increases, all companies will emit more carbon, create more value and therefore pay higher dividends. However, carbon efficient companies create more value with one unit of extra emissions than carbon intensive companies, and are less exposed to climate risk, so that the spread between green and brown company share prices will grow. Finally, the sensitivity of share price to emission efficacy appears to decrease when the proportion of green investors and their environmental stringency grow. Indeed, in a market with many green-minded investors, all companies will emit less carbon, pay lower dividends and thus have lower share prices. 

\begin{table}
\begin{center}\begin{tabular}{ccccc}
$\gamma$ & $\lambda$ & $\rho$ & $P_1$ & $P_2$\\\hline
\multicolumn{5}{c}{Impact of $\gamma$}\\\hline
$0.15$ & $0$ & $0.5$ &$1.2102\pm 0.0004$&$6.260\pm 0.011$\\
$0.3$ &$0$ & $0.5$ &$1.2112\pm 0.0003$&$6.928\pm 0.014$\\
$0.45$ &$0$ & $0.5$ &$1.2128\pm 0.0003$&$8.143\pm 0.010$\\\hline
\multicolumn{5}{c}{Impact of $\lambda$}\\\hline
$0.3$ &$0$ & $0.5$ &$1.2112\pm 0.0003$&$6.928\pm 0.014$\\
$0.3$ &$0.2$ & $0.5$ &$1.2103\pm 0.0003$&$6.834 \pm0.017$\\
$0.3$ &$0.4$ & $0.5$ &$1.2115\pm 0.0003$&$6.730\pm 0.019$\\\hline
\multicolumn{5}{c}{Impact of $\rho$}\\\hline
$0.3$ &$0.4$ & $0.0$ &$1.2112\pm 0.0003$&$6.928\pm 0.014$\\
$0.3$ &$0.4$ & $0.25$ &$1.2108\pm 0.0005$&$6.825\pm 0.020$\\
$0.3$ &$0.4$ & $0.5$ &$1.2115\pm 0.0003$&$6.730\pm 0.019$\\
$0.3$ &$0.4$ & $0.75$ &$1.2119\pm 0.0004$&$6.710 \pm 0.020$\\
$0.3$ &$0.4$ & $1.0$ &$1.2113\pm 0.0003$&$6.595 \pm 0.021$
\end{tabular}
\end{center}
\caption{Two components of the price formula \eqref{price.eq}: $P_1$ is the sensitivity to the firm value $V$, and $P_2$ is the sensitivity to the  squared emission efficacy $E^2$. The standard errors quantify the Monte Carlo error only and were computed by running the test 10 times. }
\label{price.tab}
\end{table}

\section{Conclusion}
\label{conclusion.sec}
In this paper we develop a model for the decarbonization of a large financial market, arising from an equilibrium dynamics involving companies and investors, and built using the analytical framework of mean-field games. In our model, emission schedules are stochastic and may change dynamically depending on the firm's own information, the materialization of climate risk and other market-wide or economy-wide factors. The model allows us to study the impact of climate risk and of the environmental concern of green investors on the market decarbonization dynamics in a stochastic framework: we can, for example, compute the distribution of aggregate emissions depending on the realization of common risk factors. Here, we suppose that there are two representative investors (regular and green) who have access to the same information set. In future research, the model could be extended to include many investors with varying degree of greenness and with access to information. Equilibrium dynamics will then need to be determined both at the level of companies and at the level of investors, leading to a two-stage mean-field game.

\begin{APPENDICES}

\section{Link with a $n$-firm financial market model}
\label{app:nfirm.sec} 

In this section we present a $n$-firm market model, with two investors and show that the mean field model studied in the article provides an $\varepsilon$-Nash equilibrium for the $n$-agent setting. 
This game is a leader-follower (Stackelberg) game, where the leaders are the firms.

Fix $n\geq 1$ and denote $\mathcal{N} = \left\{1,\ldots, n \right\}$.  Consider a filtered probability space $(\Omega^n,\mathcal F^n,\mathbb F^n,\mathbb P^n)$. We assume that on this probability space there exists a $d$-dimensional standard Brownian motion $B^0$ and an independent $n$-dimensional standard Brownian motion $(B^1,\dots,B^n)$. The space $(\Omega^n,\mathcal F^n,\mathbb F^n,\mathbb P^n)$ describes a financial market with $n$ companies.
To lighten the notations, we assume that $\alpha^i = \alpha$ for every $i \in \mathcal{N}$.

\paragraph{The financial market}
 We now present the two investors. Compared to the mean field situation, the only change concerns the structure of the information and thus the probability measures $\mathbb P^n$ (defined in the beginning of the appendix) and $\mathbb P^{n,g}$ defined below.
The green investor solves the optimization problem
$$
\min_{W^g} \mathbb{E}^{n,g} [e^{-\gamma^g W^g}]\quad \text{subject to}\quad \mathbb{E}^n[\xi^n W^g]\leq \mathbb{E}^n[\xi^n X^g],
$$
and the regular investor solves 
$$
\min_{W^r} \mathbb{E}^n [e^{-\gamma^r W^r}]\quad \text{subject to}\quad \mathbb{E}^n[\xi^n W^r]\leq \mathbb{E}^n[\xi^n X^r].
$$
Here, $\mathbb{E}^n$ denotes the expectation with respect to the probability measure $\mathbb P^n$, and $\mathbb{E}^{n,g}$ denotes the expectation with respect to the probability measure of the green investor $\mathbb P^{n,g}$, defined by
$$
\frac{\dd \mathbb P^{n,g}}{\dd \mathbb P^{n}} = Z,\quad Z \in \mathcal F^0_T.
$$
To simplify the analysis, we assume that $Z$ is bounded,
\begin{equation*}
    \|Z\|_{L^{\infty}(\mathcal{F}_T^0)} < +\infty.
\end{equation*}
We still assume that the density of the measure change depends only on the market-wide risk (common noise).  By an analogous reasoning to Lemma \ref{investor.lm}, the solutions to the optimization problems are of the form
$$
W^r = \mathbb{E}\left[ \xi^n X^r\right] + \frac{1}{\gamma^r} \mathbb E[\xi^n \ln \xi^n]- \frac{1}{\gamma^r} \ln \xi^n ,\qquad W^g = \mathbb{E}\left[ \xi^n X^g\right]+ \frac{1}{\gamma^g} \mathbb E[\xi^n \ln (\xi^n/Z)] - \frac{1}{\gamma^g} \ln (\xi^n/Z).
$$
The market clearing condition now writes
$$
W^r + W^g = \frac{1}{n}\sum_{i = 1}^n V^i_T,
$$
for firm values $(V^i_T)_{i \in \mathcal{N}}$, chosen in the next paragraph.
Solving for the stochastic discount factor, we find 
\begin{align} \label{eq:xi-n}
\xi^n = \frac{\exp\left(-\gamma^* \frac{1}{n}\sum_{i = 1}^n V^i_T  + \rho \ln Z\right)}{\mathbb{E}^n\left[\exp\left(-\gamma^*  \frac{1}{n}\sum_{i = 1}^n V^i_T  + \rho \ln Z\right)\right]}.
\end{align}

\paragraph{Individual firm dynamics and optimization}
We assume that \emph{the firm value} of $i$-th company follows the dynamics below:
\begin{align}
\dd  V^i_t = V^i_t(\mu^i_t \dd t + \sigma^i_t \dd B^i_t + \sigma^{i,0}_t \dd B^0_t) + e^i_t \psi^i_t \dd t,\quad V^i_0 = v_i.\label{nfirm}
\end{align}

The solution to Equation \eqref{nfirm} is given explicitly by
\begin{align*}
 V^i_t &= v_i \mathcal{E}^i_t + \mathcal{E}^i_t \int_0^t \mathcal{E}^{i,-1}_s e^i_s\psi^i_s \dd s,
\\
\mathcal{E}^i_t &= \exp\left(\int_0^t \sigma^i_s \dd B^i_s + \int_0^t \sigma^{i,0}_s \dd B^0_s  + \int_0^t \left(\mu^i_s - \frac{(\sigma^i)^2}{2} - \frac{(\sigma^{i,0})^2}{2} \right)\dd s\right).
\end{align*}
To be consistent with the limiting framework, we assume that each company has exactly $1/n$ shares outstanding. Further, we assume the existence of a complete and arbitrage-free financial market, where the shares of all firms are traded, and the stochastic discount factor $\xi^n$ is defined by \eqref{eq:xi-n}: the price of $i$-th company's stock at time $t\leq T$ is given by
\begin{align*}
S^i_t = \frac{1}{\xi^n_t}\mathbb{E}^n \left[ \left. \xi^n_T V^i_T \right | \mathcal F_t \right].
\end{align*}
In the $n$-agent setting, we consider two possible formulations for the individual firm problem, the \emph{price taker} setting and the \emph{price maker} setting. 

\paragraph{Price taker firm} 
In the price taker setting, akin to the open-loop formulation of the $n$-player games, each firm optimizes its emission strategy by assuming that the stochastic discount factor is fixed. This corresponds to the standard microeconomic equilibrium formulation. To determine its emission schedule, given a stochastic discount factor $\xi^n$, each company realizes the following program, 
\begin{align*}
\sup_{\psi^i \in L^2(\mathbb{F}^n)} J^i[\xi^n](\psi^i) ,\quad J^i[\xi^n](\psi^i) \coloneqq \mathbb{E}^n \left[\xi^n V^{i,\psi^i}_T  - \int_0^T \alpha_s c(\psi_s^i) \dd s \right],
\end{align*}
where we used the notation $\psi^{-i} = (\psi^{j})_{j \in \mathcal{N},j \neq i}$.

We call a Nash equilibrium for the $n$-player game in the price taker setting, any vector $(\xi^n, (\psi^i)_{i\in \mathcal N})$ satisfying
\begin{equation} \label{nash-eq-price-taker}
    \psi^i_t := I\left(\frac{e^i_t}{\alpha_t} \mathbb{E}^n[\xi^n \mathcal{E}^i_{t,T}|\mathcal F_t]\right),\qquad 
\xi^n = \frac{\exp\left(- \gamma^*  \frac{1}{n}\sum_{i = 1}^n V^{i,\psi^i}_T  + \rho \ln Z\right)}{\mathbb{E}^n\left[\exp\left(- \gamma^*  \frac{1}{n}\sum_{i = 1}^n V^{i,\psi^i}_T  + \rho \ln Z\right)\right]}.
\end{equation}
Using the same approach as in the mean field framework, we define the potential
\begin{equation*}
    G^n(\xi) = H(\xi) + L^n(\xi) + \chi(\mathbb{E}^n \left[\xi \right] - 1),
\end{equation*}
where 
\begin{equation*}
    L^n(\xi) = \frac{\gamma^{\star}}{n} \sum_{i = 1}^n \mathbb{E}^n\left[\xi V\mathcal{E}^i_T +\int_0^T \alpha_t c^*\left(\frac{e^i_t}{\alpha_t} \mathbb{E}^n[\xi \mathcal{E}^i_{t,T}|\mathcal F_t]\right) \dd t   \right].
\end{equation*}

\begin{proposition} \label{prop:existence-uniq-Nash-price-taker}
    There exists a unique solution $(\xi^n,(\psi^i)_{i\in \mathcal N})$ to the fixed-point problem \eqref{nash-eq-price-taker} satisfying $\xi^n = \argmin_{\xi \in \mathcal{C}} G^n(\xi)$. In addition one can verify that, as in the mean field setting,
\begin{equation*}
    \psi^i \in \mathcal{C}^n \coloneqq \{ \psi \in L_{\alpha}^2(\mathbb{F}^n), \; \|\psi\|_{L^2(\mathbb{F}^n)} \leq C,  \; \psi \geq 0, \; a.s, \, a.e.\},
\end{equation*}
for some fixed constant $C>0$, for every $i \in \mathcal{N}$.
\end{proposition}

\begin{proof}
The proof is similar to that of Theorem \ref{main.thm} with $G$ and $\mathcal{F}_T^0$ being respectively replaced by $G^n$ and $\mathcal{F}_T^{0,n}$. 
Notice that $G$ and $G^n$ only differ from $L$ being replaced by $L^n$ and $\mathbb{E}[\xi]$ being replaced by $\mathbb{E}^n[\xi]$. The latter is not difficult to handle. Then we only justify that $L^n$ is non-negative, convex, stronly convex and lower semi-continuous. The non-negativity of $L^n$ is direct and the other properties can be derived as in Lemma \ref{lemma:G}. Finally, the $L^2$ bound on the vector of strategies $(\psi^i)_{i \in \{1,\ldots,n\}}$ follows from the strong convexity of $c^*$.
\end{proof}

\begin{definition}
    Let $\varepsilon>0$. We say that a vector of strategies $(\bar{\psi}^i)_{i\in \mathcal N}$ with $\bar \psi^i \in \mathcal{C}^n$ for all $i$ is a $\varepsilon$-Nash equilibrium for the $n$-player price taker game if
    \begin{equation*}
        J^i[\bar \xi^n](\bar{\psi}^{i})  \geq \sup_{\psi \in \mathcal{C}^n} J^i[\bar \xi^n](\psi) - \varepsilon,
    \end{equation*}
    for all $i \in\mathcal N$, where 
    \begin{equation} \label{def:xi-n}
        \bar \xi^n = \frac{\exp\left(- \gamma^* \frac{1}{n}\sum_{i = 1}^n  V^{i,\bar \psi^i}_T  + \rho \ln Z\right)}{\mathbb{E}^n\left[\exp\left(- \gamma^* \frac{1}{n}\sum_{i = 1}^n V^{i,\bar \psi^i}_T  + \rho \ln Z\right)\right]}.
    \end{equation}
\end{definition}

Let $(\bar{\xi},\bar{\psi})$ be a mean field game equilibrium. Let $(\bar{\psi}^i)_{i \in \mathcal{N}}$ be conditionally i.i.d copies of $\bar{\psi}$ with respect to the common noise, and let $\bar \xi^n$ be the associated stochastic discount factor given by \eqref{def:xi-n}.

\begin{proposition} \label{prop:price-taker-eq}
     There exists $\varepsilon(n)$ such that, for every player $i$ and every admissible deviation $\psi^i\in\mathcal C^n$,
\begin{equation*}
    J^i[\bar\xi^n](\bar\psi^i)
\;\ge\;
J^i[\bar\xi^n](\psi^i)
\;-\;
\varepsilon(n),
\end{equation*}
i.e.\ the symmetric price-taker profile $\bar\psi$ forms an $\varepsilon(n)$-Nash equilibrium of the $n$-player game.
\end{proposition}

\begin{proof}
By definition of the mean field Nash equilibrium, we recall that 
\begin{equation} \label{definition:Nash-eq}
    J[\bar\xi](\bar\psi) \geq J[\bar\xi](\psi), 
\end{equation}
for all $\psi \in \mathcal{C}$.
By definition of the criterion $J^i$, we have
\begin{equation*}
     J[\bar\xi](\psi) =  J^i[\bar\xi](\psi) = J^i[\bar \xi^n](\psi) + \mathbb{E}^n\left[ V^{i,\psi}_T  (\bar \xi - \bar \xi^n) \right].
\end{equation*}
Plugging into the Nash equlibirum definition \eqref{definition:Nash-eq} we obtain
\begin{equation*}
    J^i[\bar \xi^n](\bar \psi) \geq J^i[\bar  \xi^n](\psi) + \mathbb{E}^n\left[ V^{i,\psi}_T  (\bar \xi - \xi^n) \right] - \mathbb{E}^n\left[ V^{i,\bar{\psi}}_T  (\bar \xi - \bar \xi^n) \right].
\end{equation*}
Using that  $(\bar{\psi}^i)_{i \in \mathcal{N}}$ are conditional i.i.d copies of $\bar{\psi}$ we have
\begin{equation} \label{ineq:pre-eps-nash-eq-price-taker}
    J^i[\bar \xi^n](\bar \psi^i) \geq J^i[ \bar  \xi^n](\psi) + \mathbb{E}^n\left[ V^{i,\psi}_T  (\bar \xi - \bar  \xi^n) \right] - \mathbb{E}^n\left[ V^{i,\bar{\psi}^i}_T  (\bar \xi - \bar \xi^n) \right].
\end{equation}
Now we justify that 
\begin{equation} \label{eq:SLLN}
    \lim_{n \to +\infty} \frac{1}{n}\sum_{i = 1}^n  V^{i,\bar \psi^i}_T = \mathbb{E} \left[ \left. V^{\bar \psi}_T \right \vert \mathcal{F}^0_T \right],
\end{equation}
almost surely. 
By the conditional i.i.d. assumption there exists a regular conditional probability
\begin{equation*}
    \mathbb{P}^0_{\omega}(\cdot) = \mathbb{P}(\cdot \vert \mathcal{F}_T^0)(\omega), \quad \forall \omega \in \Omega,
\end{equation*}
such that $(V^{i,\bar \psi^i}_T)_{i \in \mathcal{N}}$ are i.i.d random variables with mean 
\begin{equation*}
    m(\omega) = \mathbb{E}_{\mathbb{P}^0_{\omega}}\left[V^{i,\bar \psi^i}_T\right] = \mathbb{E}\left[\left. V^{i,\bar \psi^i}_T  \right \vert \mathcal{F}_T^0\right](\omega).
\end{equation*}
Then by the strong law of large numbers, we have that 
\begin{equation*}
    \mathbb{P}^0_{\omega}\left(A \right) = 1, \quad A = \left\{ \omega \in \Omega, \; \lim_{n \to \infty}\frac{1}{n}\sum_{i = 1}^n  V^{i,\bar \psi^i}_T(\omega) = m(\omega) \right\}.
\end{equation*}
By the tower property of the expectation we have that 
\begin{equation*}
    \mathbb{P}(A) = \mathbb{E}\left[ \mathbb{P}\left(A \vert \mathcal{F}_T^0 \right)\right] = 
    \mathbb{E}\left[ \mathbb{P}^0_{\omega}\left(A \right)\right] = 1,
\end{equation*}
 and thus \eqref{eq:SLLN} holds. 
Then by continuity of the exponential, 
$$
 \lim_{n \to +\infty}\exp\left(- \gamma^* \frac{1}{n}\sum_{i = 1}^n  V^{i,\bar \psi^i}_T  + \rho \ln Z\right) = \exp\left(- \gamma^* \mathbb{E} \left[ \left. V^{\bar \psi}_T \right \vert \mathcal{F}^0_T \right]  + \rho \ln Z\right),
$$
almost surely and thus 
$$
 \lim_{n \to +\infty}\xi^n = \bar{\xi},
$$
almost surely. By H{\"o}lder inequality, we have that
\begin{equation*}
    \left|\mathbb{E}^n\left[ V^{i,\psi}_T  (\bar \xi - \bar  \xi^n) \right]\right| \leq \|V^{i,\psi}_T\|_{L^2(\mathcal{F}_T)} \|\bar{\xi} - \bar  \xi^n\|_{L^2(\mathcal{F}_T)}
\end{equation*}
Taking the supremum on both sides with respect to the strategy, we obtain by definition of the the set $\mathcal{C}^n$,
\begin{equation*}
    \sup_{\psi \in \mathcal{C}^n}\left|\mathbb{E}^n\left[ V^{i,\psi}_T  (\bar \xi - \bar  \xi^n) \right]\right| \leq C \|\bar{\xi} - \bar  \xi^n\|_{L^2(\mathcal{F}_T)},
\end{equation*}
since there exists $C>0$ such that $\mathbb{E}^n[|V^{i,\psi}_T|^2]<C$ for all $\psi^i \in \mathcal{C}^n$ by Proposition \ref{prop:existence-uniq-Nash-price-taker}.

Since $Z$ is assumed to be bounded, the familly of random variables $|\bar \xi - \bar{\xi}^n|^2$ is uniformly integrable. Indeed, for any $\varepsilon>0$ we have
\begin{equation*}
\mathbb{E}\left[|\bar{\xi}^n|^{2+\varepsilon}\right] \leq C \mathbb{E}\left[Z^{\rho(2 + \varepsilon)}\right] < +\infty.
\end{equation*}
Finally, by uniform integrability and almost sure convergence of $(\bar{\xi}^n)_{n \in \mathcal{N}}$, we have that 
\begin{equation*}
    \lim_{n \to +\infty} \sup_{\psi \in \mathcal{C}^n} \left|\mathbb{E}^n\left[ V^{i,\psi}_T  (\bar \xi - \bar  \xi^n) \right]\right| \leq   \lim_{n \to +\infty} \|\bar{\xi} - \bar  \xi^n\|_{L^2(\mathcal{F}_T)} = 0.
\end{equation*}
Taking first the supremum and then the limit both sides in \eqref{ineq:pre-eps-nash-eq-price-taker} concludes the proof. 
\end{proof}

\paragraph*{Price maker} 
In the price maker setting, akin to the closed-loop formulation of the $n$-player games, each firm optimizes its emission strategy taking into account the impact of its strategy on the stochastic discount factor. 
In the following we denote $\psi = (\psi^1,\ldots,\psi^n)$ the vector of strategies and $\psi^{-i} = (\psi^1,\ldots,\psi^{i-1},\psi^{i+1},\ldots,\psi^n)$. Given a vector of emissions of the other companies $\psi^{-i} \in \mathcal{C}^{n-1}$, the $i$-th company realizes the following program, 
\begin{align*}
\sup_{\psi^i \in \mathcal{C}} J^i(\psi^{-i},\psi^i), \quad  J^i(\psi^{-i},\psi^i) =  \mathbb{E}^n\left[ \xi^n(\psi) V^{i,\psi^i}_T - \frac{1}{2}\int_0^T \alpha_t |\psi^i_t|^2\dd t\right],
\end{align*}
where 
$$
\xi^n(\psi^1,\dots,\psi^n) = \frac{\exp\left(-  \gamma^*  \frac{1}{n}\sum_{i = 1}^n V^{i,\psi^i}_T  + \rho \ln Z\right)}{\mathbb{E}^n\left[\exp\left(- \gamma^* \frac{1}{n}\sum_{i = 1}^n V^{i,\psi^i}_T  + \rho \ln Z\right)\right]}.
$$
We show that the mean field game equilibrium 
may approximate the equilibrium of the $n$-player game in the price maker setting. 
\begin{definition}
    Let $\varepsilon>0$. We say that a vector of strategies $\bar{\psi} \in \mathcal{C}^n$ is a $\varepsilon$-Nash equilibrium for the $n$-player price maker game if
    \begin{equation*}
        J^i(\bar{\psi}^{- i},\bar{\psi}^{i})  \geq \sup_{\psi \in \mathcal{C}} J^i(\bar{\psi}^{- i},\psi) - \varepsilon,
    \end{equation*}
    for all $i \in \mathcal N$.
\end{definition}
Let $(\bar \xi^n, (\bar \psi^i)_{i\in \mathcal N})$ be the pair defined in the previous paragraph :
$(\bar{\psi}^i)_{i \in \mathcal{N}}$ is a vector of conditionally i.i.d copies of $\bar{\psi}$ with respect to the common noise, and $\bar \xi^n$ is the associated stochastic discount factor given by \eqref{def:xi-n}.
To save space and better highlight deviations from the equilibrium of player $i \in\mathcal N$ in the following, we denote
\begin{align*}
\xi^i (\psi^i) & =  \xi^n (\bar\psi^1,\dots,\bar\psi^{i-1}, \psi^i, \bar\psi^{i+1},\dots,\bar\psi^n),\\
S^i(\psi^i) & = -\frac{\gamma^*}{n} \sum_{j\neq i} {V}^{j,\bar \psi^j}_T  -\frac{\gamma^*}{n}{V}^{i,\psi^i}_T  + \rho \ln Z,
\end{align*}
for any $\psi^i \in \mathcal{C}$.

\begin{theorem} \label{prop:vareps-nash-eq-price-maker}
    The vector $(\bar \xi^n,(\bar{\psi}^i)_{i \in \mathcal{N}})$ is a $\varepsilon(n)$-Nash equilibrium for the price maker game, where $\varepsilon(n)$ is a sequence converging to zero as $n$ tends to infinity.
\end{theorem}

\begin{proof}
By definition of the $\varepsilon(n)$-Nash equilibrium for the price taker setting, we have 
\begin{align}  \label{eq:Nash-equilibrium-price-taker}
    J^i[\bar \xi^n](\bar \psi^i) \geq J^i[\bar \xi^n](\psi^i) - \varepsilon(n),
\end{align}
for any $\psi^i \in \mathcal{C}$ and all $i \in \mathcal N$. 
Let us first recall that by definition of the criteria, we have 
\begin{equation*} 
    J^i[\bar \xi^n](\bar \psi^i) = J^i(\bar \psi^{-i},\bar \psi^i).
\end{equation*}
Then we have 
\begin{align*}
    J^i[\bar \xi^n](\psi^i) & = \mathbb E^n[\xi^i(\bar \psi^i) V^{i,\psi^i}_T  ] - \mathbb{E}^n\left[ \frac{1}{2}\int_0^T \alpha^i_t |\psi^i_t|^2\dd t\right] \\
    & = J^i(\bar \psi^{-i},\psi^i)  + 
    \mathbb E^n[ ( \xi^i(\bar \psi^i) -  \xi^i(\psi^i)) V^{i,\psi^i}_T],
\end{align*}
where we used that $\bar \xi^n = \xi^i(\bar \psi^i)$ by notational convention. Combining the last equalities with \eqref{eq:Nash-equilibrium-price-taker} yields that 
\begin{align} \label{eq:pre-nash-ineq}
    J^i(\bar \psi^{-i},\bar \psi^i) \geq J^i(\bar \psi^{-i},\psi^i)  + 
    \mathbb E^n[( \xi^i(\bar \psi^i) -  \xi^i(\psi^i)) V^{i,\psi^i}_T ] - \varepsilon(n).
\end{align}
Then it remains to estimate the second term in the right hand side. We start with 
\begin{align*}
    ( \xi^i(\bar \psi^i) -  \xi^i(\psi^i)) & = \frac{e^{S^i(\bar \psi^i)}}{\mathbb E^n[e^{S^i(\bar \psi^i)}]} - \frac{e^{S^i(\psi^i)}}{\mathbb E^n[e^{S^i(\psi^i)}]} \\
    & = \frac{e^{S^i(\bar \psi^i)} - e^{S^i(\psi^i)}}{\mathbb E^n[e^{S^n_i(\psi^i)}]} + \frac{e^{S^i(\psi^i)}}{\mathbb E^n[e^{S^i(\psi^i)}]} \left(\frac{\mathbb E^n[e^{S^i(\psi^i)}]}{\mathbb E^n[e^{S^i(\bar \psi^i)}]} - 1 \right).
\end{align*}
Using that $V^{i,\psi^i}_T \geq 0$ since $\psi \geq 0$ almost surely, almost everywhere, we have
\begin{align} \label{estiamtes:V-delta-xi}
    |\mathbb E^n[( \xi^i(\bar \psi^i) -  \xi^i(\psi^i)) V^{i,\psi^i}_T ]|  \leq  &\mathbb E^n[|\xi^i(\bar \psi^i) -  \xi^i(\psi^i)| V^{i,\psi^i}_T ] \leq a + b,
\end{align}
where 
\begin{equation*}
    a = \frac{1}{\mathbb E^n[e^{S^i(\psi^i)}]}\mathbb E^n\left[ \left| e^{S^i(\bar \psi^i)} - e^{S^i(\psi^i)} \right| V^{i, \psi^i}_T \right], \quad b = \frac{1}{\mathbb E^n[e^{S^i(\psi^i)}]} \left|\frac{\mathbb E^n[e^{S^i(\psi^i)}]}{\mathbb E^n[e^{S^i(\bar \psi^i)}]} - 1 \right| \mathbb E^n\left[e^{S^i(\psi^i)} V^{i, \psi^i}_T \right].
\end{equation*}
We now estimate the terms $a$ and $b$. We start with a lower bound on the denominator terms.
By Jensen's inequality, 
\begin{align} \label{ineq:exp-S-i-bar-below}
    \mathbb E^n[e^{S^i(\bar \psi^i)}] & \geq e^{\mathbb E^n\left[S^i(\bar \psi^i)\right]} = \exp\left(-\frac{\gamma^*}{n} \sum_{j\neq i} \mathbb E^n[{V}^{j,\bar \psi^j}_T] + \rho \mathbb E^n[\ln Z]\right), \\
    \mathbb E^n[e^{S^i(\psi^i)}] & \geq e^{\mathbb E^n\left[S^i(\psi^i)\right]} = \exp\left(-\frac{\gamma^*}{n} \sum_{j\neq i} \mathbb E^n[{V}^{j,\bar \psi^j}_T]  -\frac{\gamma^*}{n}\mathbb E^n[{V}^{i,\psi^i}_T]  + \rho \mathbb E^n[\ln Z]\right). \label{ineq:exp-S-i-below}
\end{align}
We now argue that ${V}^{i,\bar \psi^i}_T$ and ${V}^{i, \psi^i}_T$ enjoy finite second order moments, uniformly in $i \in \mathcal{N}$. Concerning ${V}^{i,\bar \psi^i}_T$, this is a direct consequence of the definition of $\mathcal{C}^n$.  
Concerning ${V}^{i, \psi^i}_T$, we can assume without loss of generality that $\psi^i \in \mathcal{C}$ is bounded in $L^2(\mathbb{F})$, uniformly in $i \in \mathcal N$ since the optimal strategy $\psi \in \mathcal{C}$ is bounded in $L^2(\mathbb{F})$ in the mean field regime. 
Then there exists $C>0$ independent on $i \in \mathcal{N}$, such that 
$\mathbb E^n[{V}^{i,\bar \psi^i}_T] \leq C$ and $\mathbb E^n[{V}^{i,\psi^i}_T] \leq C$. Then we have that 
\begin{equation*}
    \mathbb E^n[e^{S^i(\bar \psi^i)}] \geq c, \quad \mathbb E^n[e^{S^i(\psi^i)}] \geq c.
\end{equation*}
We now estimate the difference $e^{S^i(\bar \psi^i)} - e^{S^i(\psi^i)}$.
By the fundamental theorem of calculus, we have that 
\begin{align*}
    e^{S^i(\bar \psi^i)} & = e^{S^i(\psi^i)} +  \int_0^1 e^{\theta S^i(\bar \psi^i) + (1- \theta) S^i(\psi^i)} \dd \theta (S^i(\bar \psi^i) - S^i(\psi^i) \\[0.5em]
    & = e^{S^i(\psi^i)} +  \int_0^1 e^{\theta S^i(\bar \psi^i) + (1- \theta) S^i(\psi^i)} \dd \theta \frac{\gamma^*}{n}(V^{i,\bar \psi^j}_T- V^{i,\psi^i}_T  ),
\end{align*}
almost surely. 
Since ${V}_T^{i,\bar \psi^i}, {V}_T^{i,\psi^i} \geq 0$ for each $i \in \mathcal N$, the random variable $e^{\theta S^i(\bar \psi^i) + (1- \theta) S^i(\psi^i)}$ is always lower or equal to $Z^{\rho}$ which is assumed to be bounded, for any $\theta \in [0,1]$. Thus 
\begin{align} \label{ineq:exp-decrease}
    \left\vert e^{S^i(\bar \psi^i)} -  e^{S^i(\psi^i)} \right \vert \leq   \frac{C}{n} \left\vert V^{i,\bar \psi^j}_T- V^{i,\psi^i}_T \right \vert,
\end{align}
almost surely. 
Returning to the numerator of $a$, combining \eqref{ineq:exp-S-i-below}, \eqref{ineq:exp-decrease} and using Young's inequality we have
\begin{align*}
\mathbb E\left[ \left| e^{S^i(\bar \psi^i)} - e^{S^i(\psi^i)} \right| V^{i, \psi^i}_T \right] & \leq \frac{C}{n} \mathbb E\left[ \left| V^{i,\bar \psi^j}_T- V^{i,\psi^i}_T \right| V^{i, \psi^i}_T \right]
    \\[0.5em]
    & \leq \frac{C}{n}  \left(\mathbb E^n\left[|V^{i, \psi^i}_T |^2 \right] + \mathbb E^n\left[ \left|V^{i,\bar \psi^j}_T- V^{i,\psi^i}_T  \right|^2 \right] \right) \\[0.5em]
    & \leq \frac{C}{n},
\end{align*}
where the constant $C>0$ might have increased from line to line, but is independant on $i \in \mathcal{N}$.
We can now estimate $a$, 
\begin{align*}
    a \leq \frac{C}{n}.
\end{align*}
We can also estimate $b$, using \eqref{ineq:exp-S-i-bar-below} and \eqref{ineq:exp-decrease} we have
\begin{align*}
    b \leq C \left|\frac{\mathbb E^n[e^{S^i(\bar \psi^i)}]}{\mathbb E^n[e^{S^i(\psi^i)}]} - 1 \right| \leq C \left|\mathbb E^n[e^{S^i(\bar \psi^i)}-e^{S^i(\psi^i)}] \right| 
     \leq \frac{C}{n}.
\end{align*}
Finally combining the last estimates on $a$ and $b$ with (\ref{eq:pre-nash-ineq}-\ref{estiamtes:V-delta-xi}) yields 
\begin{align*}
     J^i[\bar \psi^{-i}](\bar \psi^i) \geq J^i[\bar \psi^{-i}](\psi^i) - C/n - o(1/n) -  \varepsilon(n).
\end{align*}
Taking the supremum over $\psi^i \in \mathcal{C}$ both sides yields the desired result. 
\end{proof}

We conclude this appendix, combining the two main results we obtained. 
Let $(\bar \xi^n, (\bar \psi^i)_{i\in \mathcal N})$ be the $n$-player price taker Nash equilibrium
(which exists by Proposition \ref{prop:existence-uniq-Nash-price-taker}).
\begin{corollary}
    We have that $(\bar \xi^n, (\bar \psi^i)_{i\in \mathcal N})$ is a $\varepsilon(n)$-Nash equilibrium for the $n$-player price maker game.
\end{corollary}

\begin{proof}
    By definition of the $n$-player price taker Nash equilibrium we have 
    \begin{align}  \label{eq:Nash-equilibrium-price-taker-new}
        J^i[\bar \xi^n](\bar \psi^i) \geq J^i[\bar \xi^n](\psi^i).
    \end{align}
    The conclusion follows by the same arguments as in the proof of Theorem \ref{prop:vareps-nash-eq-price-maker} replacing \eqref{eq:Nash-equilibrium-price-taker}  with \eqref{eq:Nash-equilibrium-price-taker-new}.
\end{proof}



\end{APPENDICES}

\section*{Acknowledgments.}
The research of Pierre Lavigne was supported by ADEME (Agency for Ecological Transition). The research of Peter Tankov was supported by the FIME (Finance for Energy Markets) research initiative of the Institut Europlace de Finance. Part of this research was carried out when Peter Tankov was visiting  Long-Term Investors@UniTO; the hospitality of this institution is gratefully acknowledged. We thank Olivier David Zerbib and St\'ephane Voisin for insightful comments on an earlier version of this paper.




\end{document}